%% file: main.tex
\documentclass{article}

\PassOptionsToPackage{square, numbers, compress}{natbib}
\usepackage[final]{neurips_2025}

\usepackage[utf8]{inputenc} 
\usepackage{float}
\usepackage[T1]{fontenc}    
\usepackage{hyperref}       
\usepackage{url}            
\usepackage{booktabs}       
\usepackage{amsfonts}       
\usepackage{nicefrac}       
\usepackage{microtype}      
\usepackage[dvipsnames]{xcolor}
\usepackage{graphicx}         
\usepackage{xspace}
\usepackage{subcaption}
\usepackage{amsmath}
\usepackage[linesnumbered,ruled,longend,commentsnumbered]{algorithm2e}
\usepackage{mathtools}
\usepackage{amssymb}
\usepackage{wrapfig}
\usepackage{enumitem}
\usepackage{lipsum}
\usepackage{amsthm}
\newlist{compactitem}{itemize}{1}
\setlist[compactitem,1]{label=\textbullet, left=0pt, itemsep=1pt, topsep=1pt, parsep=0pt, partopsep=0pt}

\newtheorem{theorem}{Theorem}[section]

\newtheorem{definition}{Definition}[section]

\providecommand{\ie}{\emph{i.e.,} }
\providecommand{\eg}{\emph{e.g.,} }

\providecommand{\myparab}[1]{\noindent\textbf{#1} }

\newcommand{\sysname}{FlashMoE\xspace}

\providecommand{\alltoall}{\emph{AlltoAll}\xspace}

\title{\sysname: Fast Distributed MoE in a Single Kernel}
\author{%
    Osayamen Jonathan Aimuyo\thanks{Correspondence to \texttt{osayamen@stanford.edu}} \\
    Cornell University\\
    \texttt{oja7@cornell.edu} \\
 \And
 Byungsoo Oh \\
 Cornell University \\
 \texttt{bo239@cornell.edu} \\
 \And
 Rachee Singh \\
 Cornell University \\
 \texttt{rs2293@cornell.edu} \\
}
\begin{document}
    \maketitle
    \input{content/abstract}
    \input{content/introduction}
    \input{content/motivation}
    \input{content/method-v2}
    \input{content/evaluation-v2}
    \input{content/limitations}

    \input{content/conclusion}

    \input{content/acks}
    \clearpage
    \bibliography{main}
    \clearpage
    \input{content/checklist}
    \clearpage
    \appendix
    \input{content/appx/motivation}
    \input{content/related}

    \input{content/appx/proofs}
    \clearpage
    \input{content/appx/task}
    \clearpage
    \clearpage
    \input{content/appx/actors}
    \input{content/appx/impl}
    \input{content/appx/fp16_t}

    \clearpage
\end{document}

%% file: content/abstract.tex
\begin{abstract}
    The computational sparsity of Mixture-of-Experts (MoE) models enables sub-linear growth in compute cost as
    model size increases, thus offering a scalable path to training massive neural networks.
    However, existing implementations suffer from low GPU utilization, significant latency overhead,
    and a fundamental inability to leverage task locality,
    primarily due to CPU-managed scheduling, host-initiated communication, and frequent kernel launches.
    To overcome these limitations, we develop~\sysname,
    a fully GPU-resident MoE operator that fuses expert computation and inter-GPU communication
    into a single persistent GPU kernel.
    \sysname enables fine-grained pipelining of dispatch, compute, and combine phases,
    eliminating launch overheads and reducing idle gaps.
    Unlike existing work, \sysname obviates bulk-synchronous collectives
    for one-sided, device-initiated, inter-GPU (R)DMA transfers, thus unlocking
    payload efficiency, where we eliminate bloated or redundant network
    payloads in sparsely activated layers.
    When evaluated on an 8-H100 GPU node with MoE models having up to 128 experts and 16K token sequences,
    \sysname achieves up to \textbf{9}$\times$ higher GPU utilization, \textbf{6}$\times$ lower latency,
    \textbf{5.7}$\times$ higher throughput, and \textbf{4}$\times$ better overlap efficiency
    compared to state-of-the-art baselines—despite using FP32 while baselines use FP16.
    \sysname shows that principled GPU kernel-hardware co-design
    is key to unlocking the performance ceiling of large-scale distributed ML\@.
    We provide code at \url{https://github.com/osayamenja/FlashMoE}.
\begin{figure}[!h]
    \centering
    \begin{subfigure}{0.38\textwidth}
        \centering
        \includegraphics[width=\linewidth, keepaspectratio]{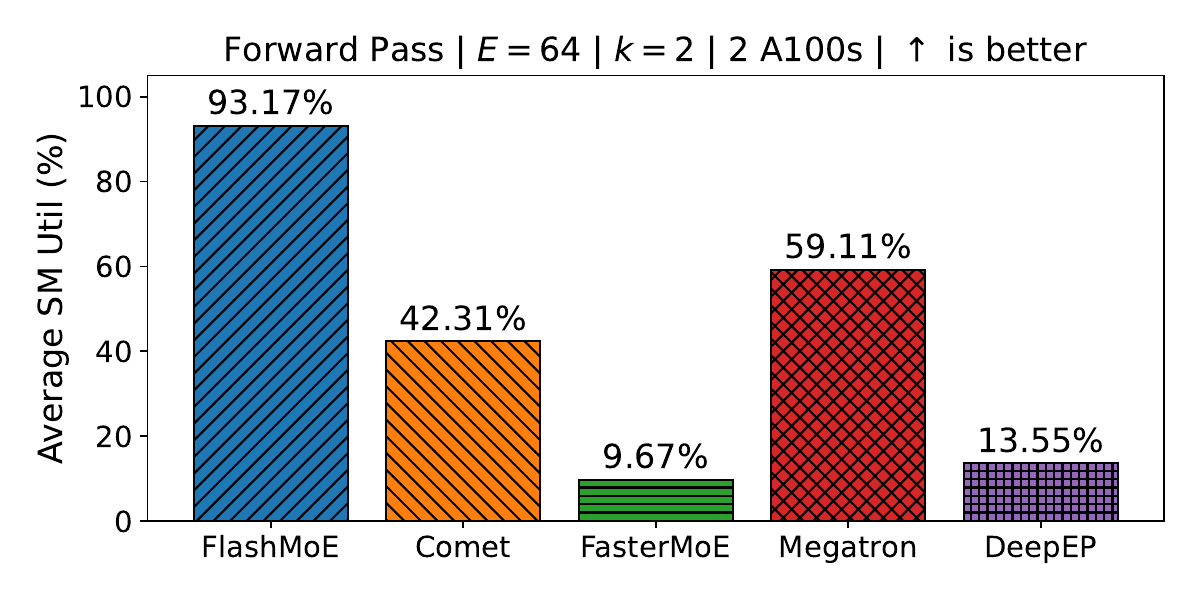}
        \caption{GPU SM Utilization}
        \label{sub:e}
    \end{subfigure}
    \begin{subfigure}{0.38\textwidth}
        \centering
        \includegraphics[width=\linewidth, keepaspectratio]{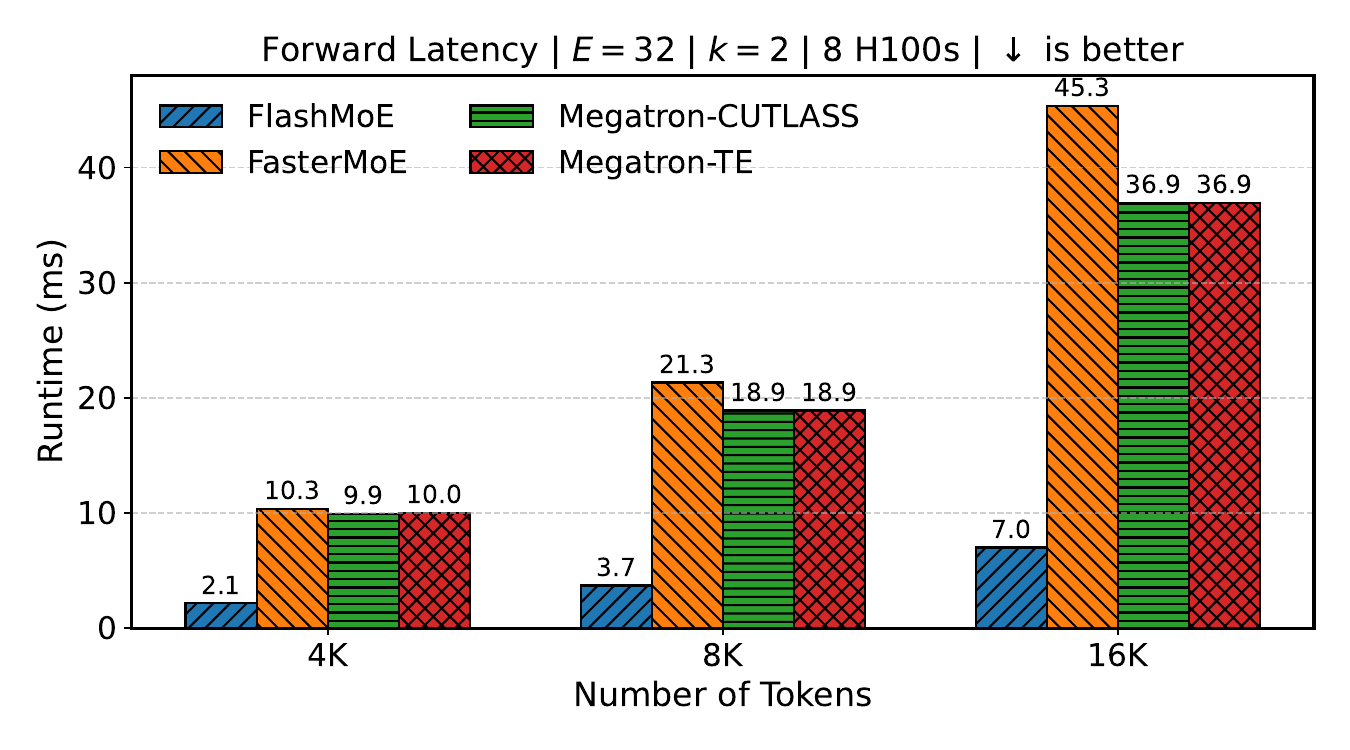}
        \caption{Scaling Tokens}
        \label{sub:r}
    \end{subfigure}
    \begin{subfigure}{0.38\textwidth}
        \centering
        \includegraphics[width=\linewidth, keepaspectratio]{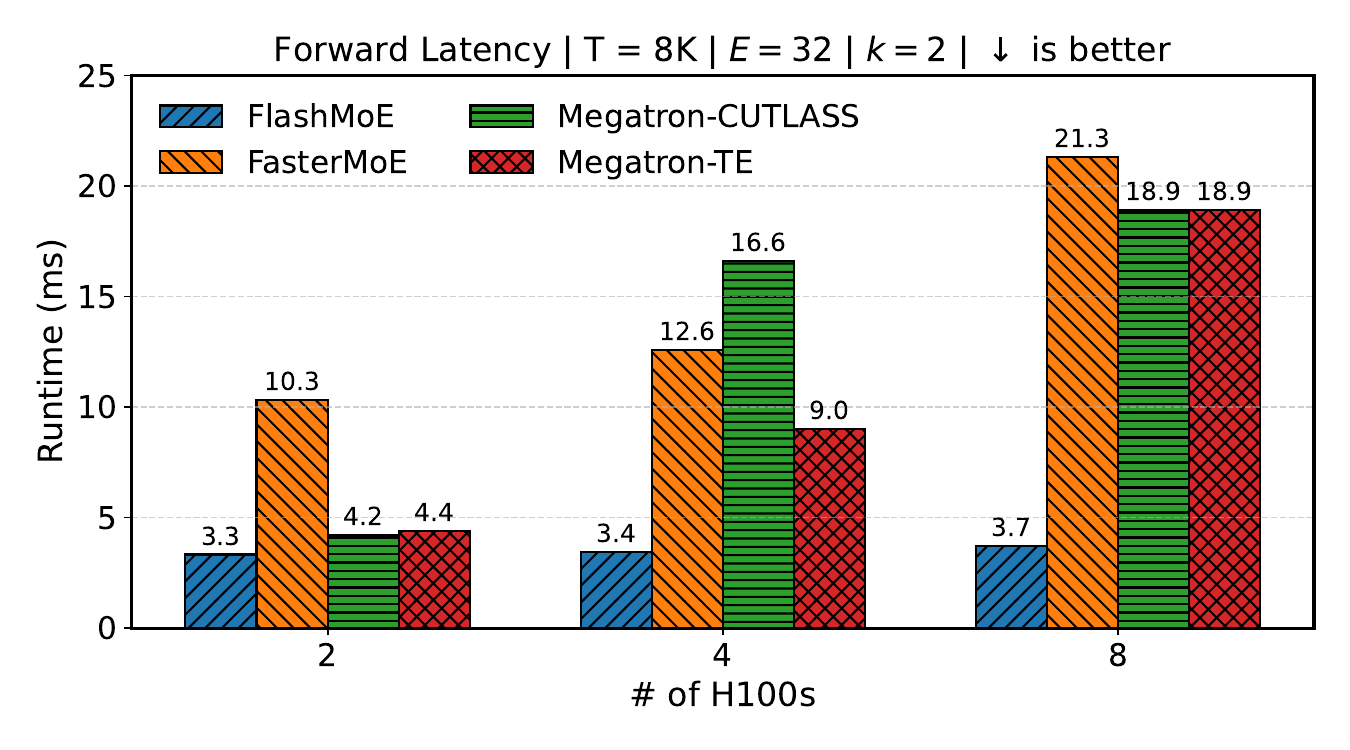}
        \caption{Weak Scaling across GPUs}
        \label{sub:e1}
    \end{subfigure}
    \begin{subfigure}{0.38\textwidth}
        \centering
        \includegraphics[width=\linewidth, keepaspectratio]{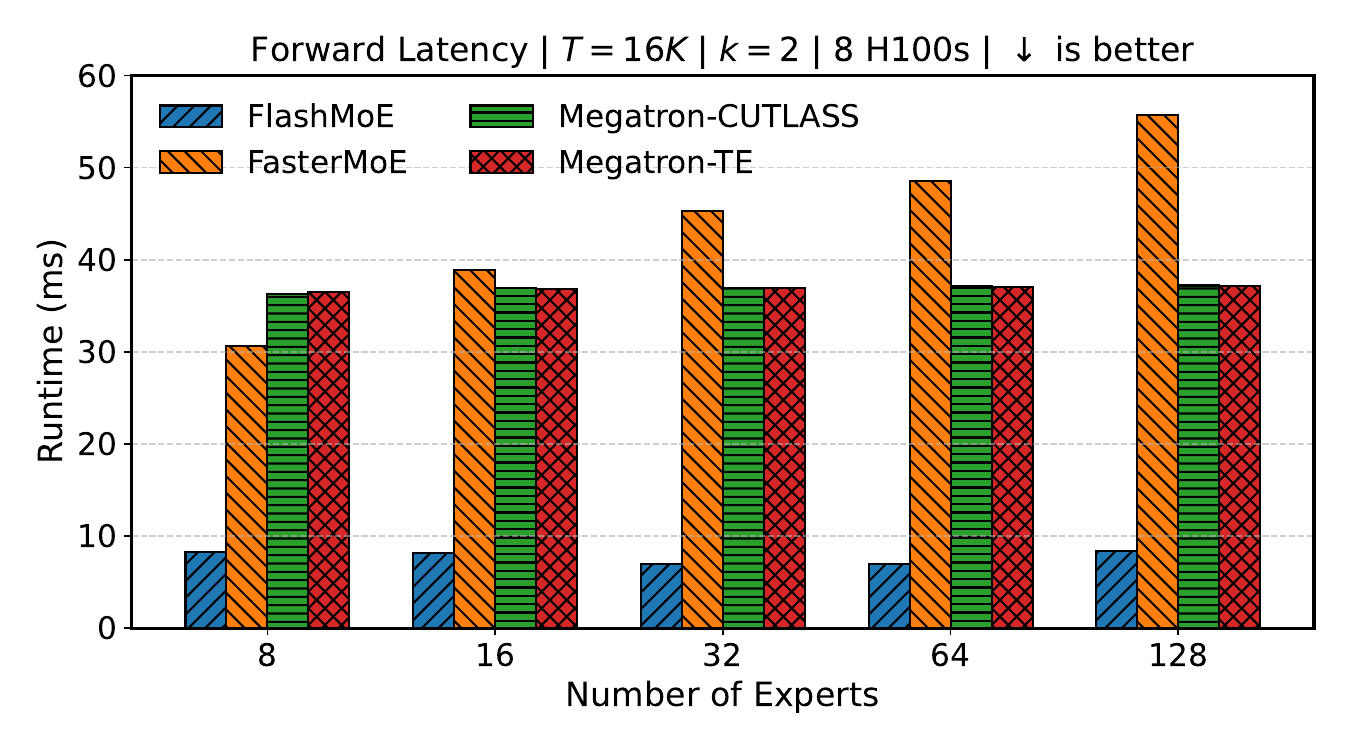}
        \caption{Across Experts}
        \label{sub:r2}
    \end{subfigure}
    \caption{\sysname performance.}
    \label{fig:str}
\end{figure}
\end{abstract}

%% file: content/introduction.tex
\section{Introduction}\label{sec:introduction}

\begin{wrapfigure}{r}{0.5\textwidth}
    \centering
    \includegraphics[width=0.51\textwidth, keepaspectratio]{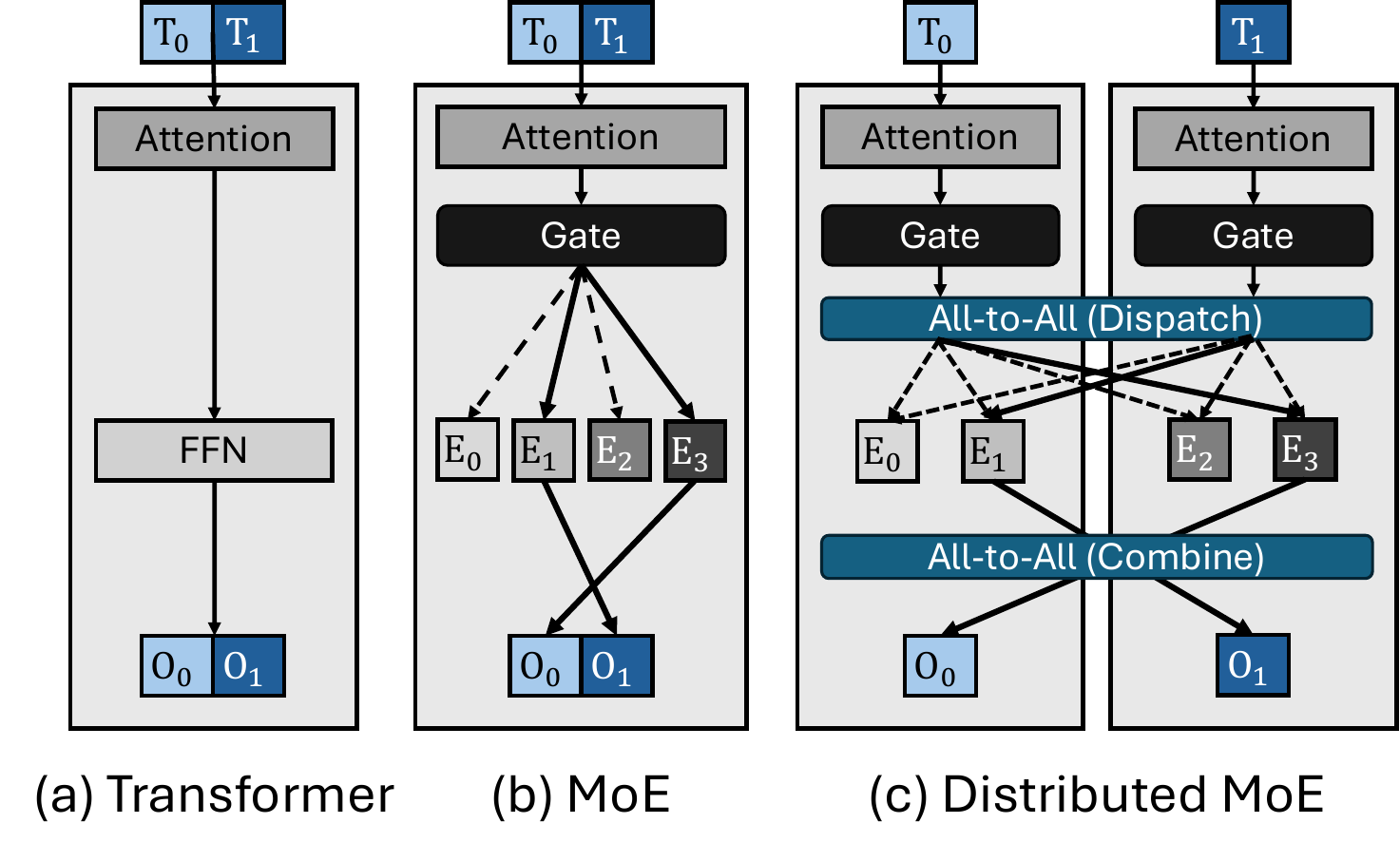}
    \caption{Transformer blocks (a) without MoE, (b) with MoE, and (c) with distributed MoE and expert parallelism.
    \texttt{T}, \texttt{E}, and \texttt{O} represent input tokens, experts, and output activations, respectively.}
    \label{fig:bg:moe}
    \vspace{-10pt}
\end{wrapfigure}

State-of-the-art large language models (LLMs)~\cite{deepep, llama4, dbrx, arctic, openai2025gptoss} have adopted the Mixture-of-Experts (MoE) architecture for its computational efficiency and strong performance across a range of tasks. The traditional Transformer block consists of a self-attention module followed by a dense feed-forward network (FFN)~\cite{NIPS2017_3f5ee243}. In contrast, MoE architectures replace this single FFN (Figure~\ref{fig:bg:moe}(a)) with many identically sized FFNs, known as experts (Figure~\ref{fig:bg:moe}(b)). A trainable neural network, known as a gate function, sparsely activates these experts by dynamically routing input tokens to the experts selected at runtime. This increase in model parameters due to more FFNs improves model quality without the corresponding increase in computational cost.

\myparab{Communication overheads in MoE.}  
As MoE model sizes grow, GPU memory constraints prevent hosting all experts on a single device. The standard practice is to distribute experts across multiple GPUs using expert parallelism (EP), which requires token routing via many-to-many communication primitives like \alltoall~\cite{deepep, arctic, dbrx, 10.1145/3577193.3593704} (Figure~\ref{fig:bg:moe}(c)). Another round of \alltoall communication restores the permuted tokens processed by experts to their original order in the sequence. \alltoall communication is challenging to optimize on GPU networks and is highly sensitive to straggler delays --- a phenomenon where a single straggler GPU delays all others from making progress~\cite{stragglar}. These communication operations can account for 68\% of the total runtime~\cite{10.1145/3603269.3604869, MLSYS2024_339caf45}, causing GPUs to remain idle (Figure~\ref{fig:intro}, top left).


\myparab{Kernel launch overheads in MoE.} 
To mitigate these communication bottlenecks, recent work pipelines computation with communication kernels (Figure~\ref{fig:intro}, left middle). However, the effectiveness of these solutions is limited by the overhead of launching many kernels from the CPU. 
Specifically, existing implementations~\cite{pmlr-v162-rajbhandari22a, comet, megatron, fastermoe} launch a large number of kernels per a single layer pass (see Table~\ref{tab:gpuOps}). Frequent kernel launches negatively affect performance by:
(1) creating non-deterministic kernel start times across GPUs, exacerbating straggler issues; (2) introducing unnecessary synchronization points, causing GPUs to wait on peers or the CPU before proceeding; and (3) incurring repeated global memory round trips at kernel boundaries. Although CUDA graphs~\cite{cuda_graphs_nvidia_blog} can partially mitigate the first issue in static workloads, they are incompatible with MoE's dynamic expert routing patterns. Addressing the remaining issues requires novel solutions,
which we provide in this work through complete kernel fusion and asynchronous device-initiated communication.

\begin{figure}[!ht]
    \centering
    \includegraphics[width=0.98\textwidth, keepaspectratio]{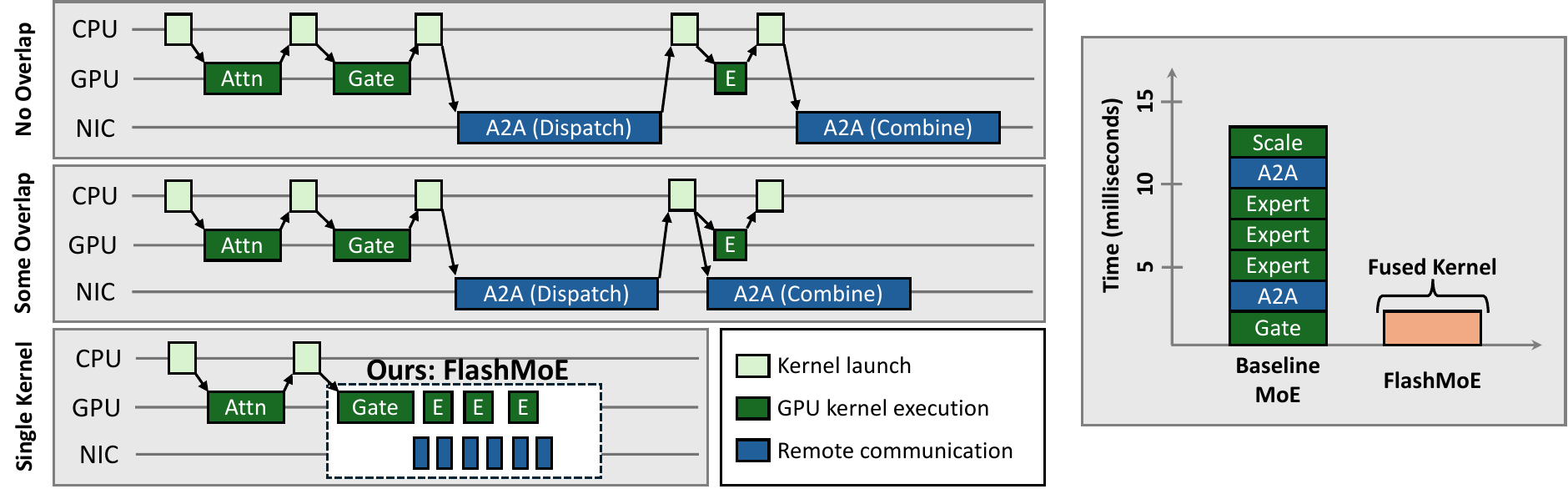}
    \caption{Comparing \sysname with state-of-the-art techniques that either do not overlap communication and computation (left, top) or do some overlap (left, middle). \sysname is a persistent kernel that fuses all computation and communication of the MoE operator (left, bottom). \sysname implements device-initiated computation (gate, expert FFN, scale) and communication tasks (right).}
    \label{fig:intro}
    \vspace{-10pt}
\end{figure}

\myparab{Our Contributions: distributed MoE in a single kernel.}
To overcome these fundamental inefficiencies in state-of-the-art MoE models, we develop \sysname a megakernel that integrates all MoE computation and communication tasks into a single persistent GPU kernel
\ie a kernel that remains active for the entirety of the MoE operator (Figure~\ref{fig:intro} bottom left). Instead of multiple kernel launches coordinated by the CPU, \sysname requires launching only one kernel,
significantly reducing the involvement of the CPU. Within the fused kernel, \sysname implements a reactive programming model to achieve
fine-grained parallelism and loosely coupled, non-blocking execution among tens of thousands of GPU threads.

\begin{wraptable}{r}{0.5\textwidth}
    \centering
    \small
    \renewcommand{\arraystretch}{0.9}
    \begin{tabular}{@{}lc@{}}
        \toprule
        \textbf{MoE Implementation} & \textbf{GPU Ops} \\ \midrule
        \sysname (this work) & 1 \\
        COMET~\cite{comet} & 33 \\
        Megatron-LM CUTLASS~\cite{megatron, 10.1145/3458817.3476209} & 85 \\
        Megatron-LM TE~\cite{megatron, 10.1145/3458817.3476209} & 261 \\
        Megatron-LM + DeepEP~\cite{deepep} & 432 \\
        DeepSpeedMoE~\cite{pmlr-v162-rajbhandari22a} & 550 \\
        \bottomrule
    \end{tabular}
    \vspace{1mm}
    \caption{
    We report number of GPU operations launched by MoE implementations by profiling with Nsight Systems~\cite{nsight-metrics}. We count operations in a single MoE layer (Gate $\rightarrow$ Dispatch $\rightarrow$ Expert $\rightarrow$ Combine) on 2 A100 GPUs, where each GPU has 32 experts. \sysname is the first to fully fuse the distributed MoE layer into a single GPU kernel.}
    \label{tab:gpuOps}
\end{wraptable}

\myparab{In-kernel Block scheduling and Tile parallelism.}
\sysname implements \emph{tile-level parallelism},
meaning it partitions input token matrices into smaller, independent units called \emph{tiles}, which are processed by blocks but managed (scheduled and constructed) by warps. We specialize every thread block, except one, as \emph{processors} to perform compute.
In addition, we designate a dedicated Operating System (OS) block (4 warps) to perform administrative tasks of
(1) scheduling computational work to processors (\emph{scheduler}), and (2) decoding computational tasks from messages received from other GPUs (\emph{subscriber}).
This design allows \sysname to dynamically assign tasks to GPU blocks based on \emph{readiness}, ensuring that no GPU SM remains idle throughout the lifetime of the MoE operator. \sysname selects tile dimensions to maximize GPU arithmetic intensity while benefitting from a high-degree of parallelism.

\myparab{Asynchronous and payload-efficient communication.}
By redesigning the MoE operator from the ground up,
\sysname resolves fundamental inefficiencies inherent in the conventional MoE execution pipeline. One notable inefficiency is token padding during communication.
To simplify programming complexity and due to symmetry constraints of collective communication APIs, existing implementations have to zero-pad token payloads to match predefined buffer sizes. This occurs when tokens are asymmetrically routed to experts, resulting in GPUs receiving much less than the expected capacity.
However, these null payloads waste communication bandwidth, bloat data transfer latency and may lead to
unnecessary computations on null matrices in some implementations. \sysname introduces \emph{payload-efficient} communication by sending non-padded tokens only to GPUs with actively selected experts, conserving both communication and computational resources.

\myparab{Technical challenges.}
Realizing the single-kernel design of \sysname required
solving several technical challenges to achieve high performance: (1) lightweight computational dependency management; (2) navigating optimal SM occupancy configurations; (3) implementing in-device BLAS operations; (4) minimizing inter- and intra-device synchronization overheads; (5) implementing transfer-awareness by leveraging DMA over Unified Virtual Addressing (UVA) when available. In addressing these challenges, \sysname's design presents a radical departure from traditional synchronous \alltoall collectives, where GPUs exhibit significant idle time during layer execution. For device-initiated communication, \sysname uses NVSHMEM~\cite{nvshm} to establish a global address space across all GPUs for Direct Memory Access (DMA) communication. For in-device BLAS, \sysname develops custom high-performance GEMM operations via CUTLASS~\cite{Thakkar_CUTLASS_2023}.

\myparab{Results.}
Our evaluations show that \sysname achieves \textbf{6}$\times$ latency speedup,
\textbf{9}$\times$ higher GPU utilization, \textbf{4}$\times$ better weak scaling efficiency and \textbf{5.7}$\times$
increased throughput compared to state-of-the-art implementations.
We project these performance gains becoming even better in multi-node scenarios,
where inter-node communication occurs using lower bandwidth inter-node links (\eg RDMA, Infiniband).

%% file: content/motivation.tex
\section{Motivation}\label{sec:motivation}
\begin{figure}[!h]
    \centering
    \begin{subfigure}{0.425\textwidth}
        \centering
        \includegraphics[width=\linewidth, keepaspectratio]{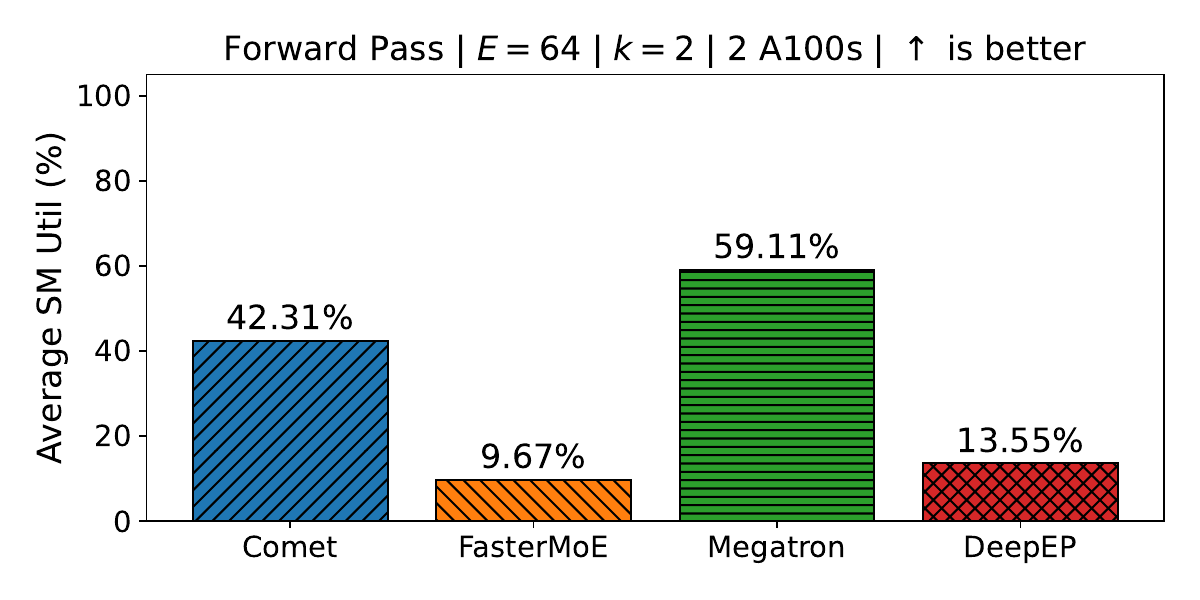}
        \caption{GPU SM Utilization across baselines}
        \label{sub:util}
    \end{subfigure}
    \begin{subfigure}{0.425\textwidth}
        \centering
        \includegraphics[width=\textwidth, keepaspectratio]{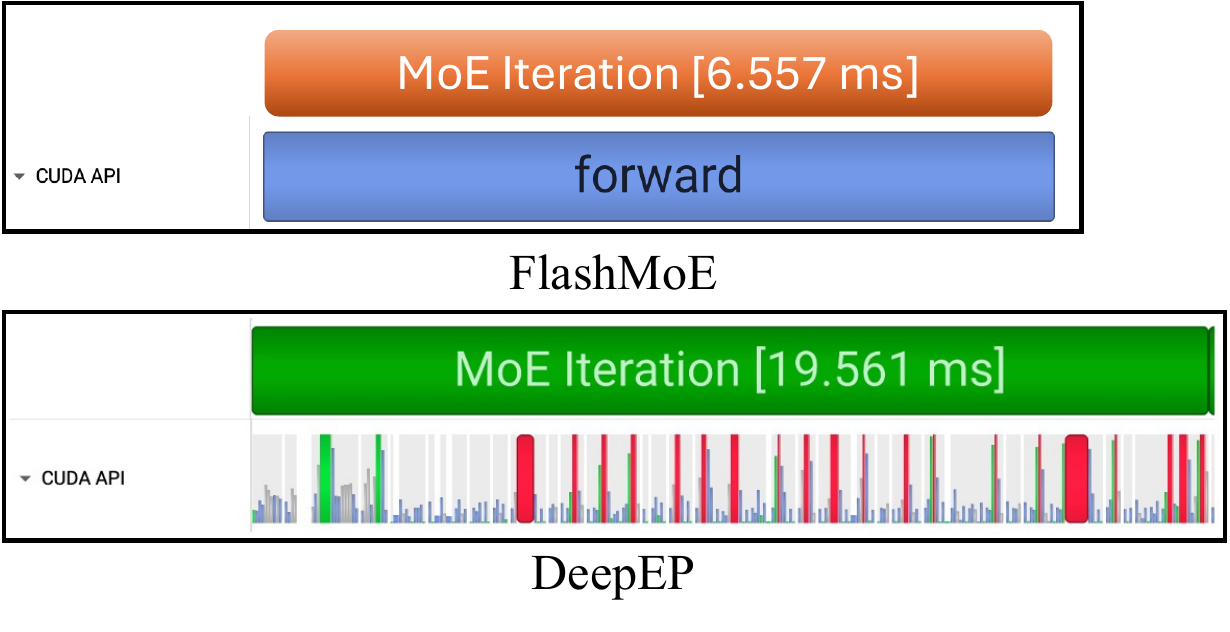}
        \caption{Kernel Launch overhead (CUDA API row)}
        \label{sub:launch}
    \end{subfigure}
    \caption{\ref{sub:util} shows GPU utilization averaged across 100 MoE forward passes on 2 A100s with
    300 GB/s unidirectional bandwidth, where we observe up to 90\% idle time, due to kernel launch gaps and
    non-overlapping communication.}
    \label{fig:kl}
\end{figure}
\myparab{Synchronous Communication.} \emph{AlltoAll} or \emph{AllGather} communication as currently used in MoE frameworks
is a \emph{synchronous} collective operation, whose completion requires the participation of all involved GPUs.
Here, disparities in processing speeds or kernel scheduling
among GPUs induce a straggler effect detrimental (Figure~\ref{fig:overlap}) to (1) the collective operation's performance and (2)
E2E performance, as stalled GPUs cannot proceed to downstream dependent or independent tasks until the collective terminates.
We expound on this problem in~\S\ref{sec:motivation-appx}.

\myparab{Kernel Launch Overhead.} We compare the kernel launch overheads between \sysname and existing baselines.
Table~\ref{tab:gpuOps} shows the number of kernel launches during a single forward pass: \sysname launches exactly one persistent kernel, while the baselines launch up to 550 short-lived kernels to perform the same computation.
Figure~\ref{fig:kl} provides a visual comparison using CUDA API traces captured by NSight Systems, illustrating the difference between \sysname and DeepEP.
DeepEP exhibits many small CUDA API calls, with frequent stalls between individual operators, leading to increased GPU idle time (Figure~\ref{sub:util}).
In contrast, \sysname maintains high GPU utilization by avoiding launch overhead and synchronization gaps—achieving \textbf{93.17}\%
GPU utilization compared to 14\% for DeepEP. See \S\ref{sec:evaluation} for experimental results and \S\ref{sec:related}
for a discussion of related work.

%% file: content/method-v2.tex
\section{Fused MoE Kernel Design}\label{sec:method}
\begin{wrapfigure}{r}{0.6\textwidth}
    \centering
    \includegraphics[width=0.6\textwidth, keepaspectratio]{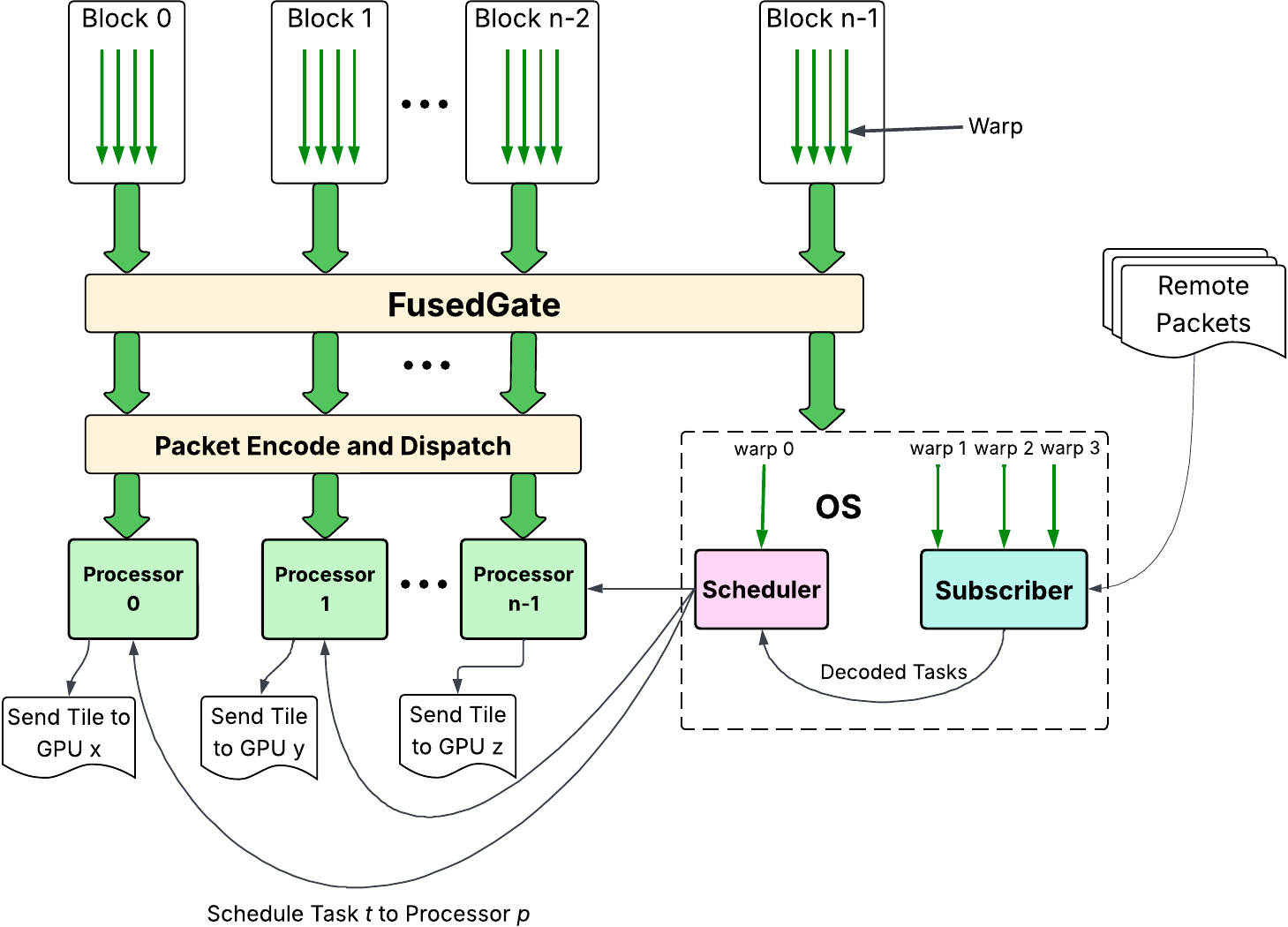}
    \caption{\emph{\sysname Fused Kernel}}
    \label{fig:fusedK}
\end{wrapfigure}
Modern distributed MoE systems suffer from two limitations: (1) frequent many-to-many
(\emph{AlltoAll or AllGather}) collectives on the critical path, and
(2) significant overhead from repeated kernel launches.
We address these in \sysname, a fully fused MoE operator implemented
as a single persistent GPU kernel.
Unlike previous approaches~\cite{comet, deepep, pmlr-v162-rajbhandari22a, megatron, MLSYS2023_5616d34c,
    MLSYS2024_339caf45, 10.1145/3503221.3508418, 10.1145/3588964, 10.1145/3627703.3650083, 10.1145/3710848.3710868,
    NEURIPS2022_67d57c32},
\sysname is the first solution to implement a \emph{completely fused Distributed MoE kernel},
eliminating kernel launch overhead entirely by requiring only a single kernel launch (see Table~\ref{tab:gpuOps}).
\SetKwInput{KwRequire}{Require}
\SetKwInput{KwResult}{Result}
\SetKwInput{KwInput}{Input}
\SetKw{kwAnd}{and}
\SetKw{kwOr}{or}
\SetKw{kwTrue}{True}
\SetKw{kwFalse}{False}
\begin{algorithm}[!h]
    \small
    \DontPrintSemicolon
    \caption{~\emph{\sysname Distributed MoE Fused Kernel}}\label{alg:one}
    \KwInput{$A, O \in \mathbb{R}^{S\times H},\; X \in \mathbb{R}^{E\times H \times D},\; N$}
    \Begin{
        $T_{\phi}, G_{\phi} \gets \mathbf{FusedGate}(A)$\;
        \eIf{$\text{blockId} + 1 < N$}{
            $\mathbf{Dispatch}(T_{\phi}, A)$\;
            processor::start()\;
        }{
            \eIf{$warpID == 0$}{
                scheduler::start()\;
            }{
                subscriber::start($T_{\phi}$, $G_{\phi}$, $O$, $X$)\;
            }
        }
    }
\end{algorithm}
\begin{figure}[!ht]
    \centering
    \includegraphics[width=\linewidth]{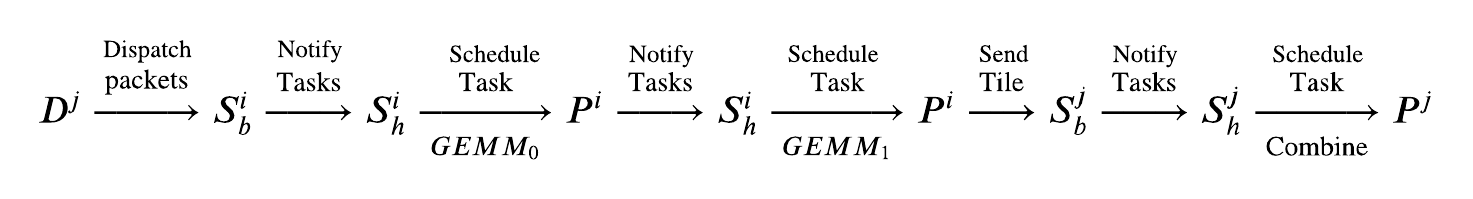}
    \caption{\emph{DMoE Functional Dependencies Expressed as a Chain of Actor Interactions}.
    We denote $S_b$, $S_h$, and $P$ as the
    Subscriber, Scheduler and Processor actors, respectively. For any actor $a \in \{S_b,\>S_h,\>P\}$,
        $a^i$ identifies an actor on GPU $i$. We define $D^j_i$ as the operator,
        where GPU $j$ dispatches packets of tiles to GPU $i$.
        This diagram expresses task dependencies at the granularity of a tile, namely
        $GEMM_0$, $GEMM_1$, combine and communication produce an output tile.
        Notifications occur as signals propagated through shared memory (subscriber $\leftrightarrow$ scheduler) or
        global memory (scheduler $\leftrightarrow$ processor or inter-GPU communication). Note one-sided
        inter-GPU transfers (packet or single tile) are \emph{coupled} with a signal to
        notify $S_b^j$ on the receiving GPU $j$ of the message's delivery.
    }
    \label{fig:actors}
    \vspace{-4mm}
\end{figure}

\myparab{Actor-based model.}
The design of \sysname is based on the actor model of concurrent
computation~\cite{agha:85, 10.5555/1624775.1624804, Greif:75}.
We implement this model by specializing GPU thread blocks and warps into three distinct actor roles:
(1) \textbf{Processor} (\S\ref{subsec:processor}), (2) \textbf{Subscriber} (\S\ref{subsec:subscriber}),
and (3) \textbf{Scheduler}(\S\ref{subsec:scheduler}).
The Processor performs compute (GEMMs and element-wise operations) and tile communication.
We use CUTLASS~\cite{Thakkar_CUTLASS_2023} as the underlying infrastructure for high-performance
BLAS routines and NVSHMEM for kernel-initiated communication~\cite{nvshm}.
The Subscriber and Scheduler perform administrative functions.
Specifically, the Scheduler assigns computational tasks to available thread blocks.
Our key innovation is making the Scheduler both \emph{multithreaded},
enabling high scheduling throughput, and \emph{work-conserving}, ensuring consistently high GPU SM utilization.
On the other hand, the Subscriber decodes \emph{tile packets} from peer GPUs to task descriptors
(\S\ref{subsec:task-abstraction-for-computation}).
Of the $N$ thread blocks on a GPU, we specialize $N-1$ to adopt the \textbf{Processor} role.
We specialize the last block as the Operating System (OS).
Within this block, we specialize three warps for the \textbf{Subscriber} role and
one warp for the \textbf{Scheduler} role.
This split of thread blocks across actors is intentional: our goal is to use few resources for administrative
tasks while reserving bulk of the resources for performing MoE computation tasks.
Figure~\ref{fig:fusedK} summarizes the
\sysname architecture and its constituent actors, while Algorithm~\ref{alg:one} gives a very close translation of the
system in code.
Note that $A \in \mathbb{R}^{S \times H}$ is the input token matrix;
$O \in \mathbb{R}^{S \times H}$ the output matrix;
and $X \in \mathbb{R}^{E\times H \times D}$ is a 3-D tensor of expert weights,
where $E$ denotes the number of local experts for the executing GPU, $H$ is the embedding dimension,
$D$ is the FFN intermediate dimension and $S$ is the sequence length.
$T_{\phi} \in \left(\mathbb{N}\times\mathbb{R}\right)^{E \times C}$
is a routing table data structure, where $T_{\phi}\left( e, c\right) = (i, w)$ indicates that token $i$ at slot $c$
dispatches to expert $e$. $w$ is the combine weight (Equation~\ref{eq:combine1}) and $C$ is expert capacity.
The tuple structure of $T_{\phi}$ is an implementation detail. $G_{\phi} \in \mathbb{R}^{S \times E}$ captures
the affinity scores produced by the gate (Equation~\ref{eq:combine2}).
\myparab{Inter-actor interactions in \sysname.}
\sysname decomposes MoE computation and communication at the granularity of a tile, a statically sized partition of a tensor,
to achieve parallel execution and efficient overlap of tasks.
Each tile maps to a discrete unit of work encapsulated by a \emph{task descriptor}.
The \textbf{Subscriber} decodes these task descriptors from the remote tile packets it receives.
Concurrently, the \textbf{Scheduler} receives notifications about available tasks and dispatches them for execution
to \textbf{Processor} actors that perform computations defined by these tasks,
namely the feed-forward network (FFN) and expert-combine operations.
Figure~\ref{fig:actors} show the chain of actor interactions, demonstrating how \sysname
enforces DMoE functional dependencies.

\myparab{Determining tile dimensions in \sysname.}
Selecting appropriate tile dimensions in \sysname is crucial to ensure efficient GPU utilization.
An undersized tile underutilizes the GPU,
while excessively large tiles create register pressure,
causing performance-degrading register spills to local memory.
After careful parameter sweeps,
we choose tile dimensions of (128, 64).
Our key insights are: increasing tile width significantly raises the register usage per thread,
potentially triggering costly spills;
increasing tile height without adjusting thread count increases workload per thread, harming performance.
Raising the thread count per block beyond our fixed value of 128 threads reduces the number of concurrent blocks,
negatively affecting SM occupancy.
Larger thread-block sizes also increase overhead from intra-block synchronization (\emph{\_\_syncthreads()} barriers),
further degrading performance.
Thus, our chosen tile dimensions balance register usage, shared-memory constraints,
and GPU occupancy to deliver optimal performance.

\subsection{Task Abstraction for Computation}\label{subsec:task-abstraction-for-computation}

\myparab{Computational operators.}
The FFN operator is a standard position-wise feed-forward network widely used in Transformer architectures~\cite{NIPS2017_3f5ee243}, composed of two linear transformations separated by a nonlinear activation $\phi$ (e.g., GELU or ReLU):

\begin{equation}\label{eq:ffn}
    \small
\textrm{FFN}(x) = W_2 \cdot \phi(x W_1 + b_1) + b_2
\end{equation}

Here, $W_1$ and $W_2$ represent learnable weight matrices, and $b_1$ and $b_2$ are biases.
The expert-combine operation, used in architectures like GShard~\cite{DBLP:conf/iclr/LepikhinLXCFHKS21} and DeepSeek~\cite{deepep}, merges outputs from multiple experts by computing a weighted combination based on their affinity scores:
\begin{equation}\label{eq:combine1}\small
\mathcal{C}_i = \sum\limits_{j = 1}^k g_{i, e}
\end{equation}
\begin{equation}\label{eq:combine2}\small
\mathbf{h}_i = \sum\limits_{j = 1}^k \frac{g_{i, e}}{\mathcal{C}_i}\cdot \mathbf{h}_i^k
\end{equation}

In these equations, $i \in {0, S - 1}$ represents an input token index, $e = E_{i,k}$ identifies the $k$-th expert selected for token $i$, and $g_{i,e}$ is the affinity score indicating how relevant expert $e$ is for token $i$.

\myparab{Unified task abstraction.}
We unify the FFN and combine operations under a common abstraction called a \emph{task}. Tasks provide a uniform interface for communicating tile-level work among Subscribers, Schedulers, and Processors. Formally, a task descriptor $t \in \mathcal{T}$ is defined as a tuple:
\[\small
    t = (\mathcal{M}, \star, \phi)
\]

where $\mathcal{M}$ is a set of metadata (\eg  device ID, tile index), $\star$ is a binary tensor operation (specifically, matrix multiplication $\cdot$ or Hadamard product $\odot$), and $\phi$ is an element-wise activation function (e.g., ReLU or identity). 

We define a task $t$ operating on input tensors $A$, $B$, $D$, producing output tensor $C$, as follows:
\begin{equation}\label{eq:task_def}\small
    \mathcal{F}_t(A, B, C, D) \coloneqq C \gets \phi\left(A \star_t B + D\right)
\end{equation}

The operator $\star_t$ (instantiated from $\star$) may behave differently depending on the task metadata $\mathcal{M}$, and the result of $A \star_t B$ is accumulated into $D$. We provide an example of task metadata in~\S\ref{sec:task-implementation}.

In practice, we implement each task defined by Equation~\ref{eq:task_def} as a \emph{single fused} \verb|__device__|
decorated function which the \textbf{Processor} (Algorithm \ref{alg:processor}) invokes at runtime.
Fusion for $t$ entails applying $\phi$ and the succeeding addition operation to registers
storing the results of the binary operator $\star_t$.
To illustrate its flexibility, we show how the FFN and expert-combine operations can be expressed
using this task framework.
Note that we omit the matrix multiplication symbol ($\cdot$) for simplicity.
Also, $\phi_1$ can be any activation function, while $\phi_2$ is the identity function.
The FFN is expressed as:
\begin{gather*}\small
    t_1 = (\mathcal{M}, \cdot, \phi_1), \quad t_2 = (\mathcal{M}, \cdot, \phi_2), \\ \small
    \mathcal{F}_{t_1}(A, B_1, C_1, D_1) \coloneqq C_1 \gets \phi_1\left(A B_1 + D_1\right), \\ \small
    \mathcal{F}_{t_2}(C_1, B_2, C_2, D_2) \coloneqq C_2 \gets \phi_2\left(C_1 B_2 + D_2\right).
\end{gather*}
Whereas, the expert-combine operation is formalized as:
\begin{gather*}\small
    t_3 = (\mathcal{M}, \odot, \phi_2), \\ \small
    \mathcal{F}_{t_3}(A, S, C, C) \coloneqq C \gets \phi_2\left(A \odot S + C\right).
\end{gather*}
\subsection{Symmetric Tensor Layout for Inter-GPU Communication}\label{subsec:symmetric-tensor-layout}
\begin{figure}[!ht]
    \centering
    \begin{subfigure}{0.55\textwidth}
        \centering
        \includegraphics[width=\linewidth, keepaspectratio]{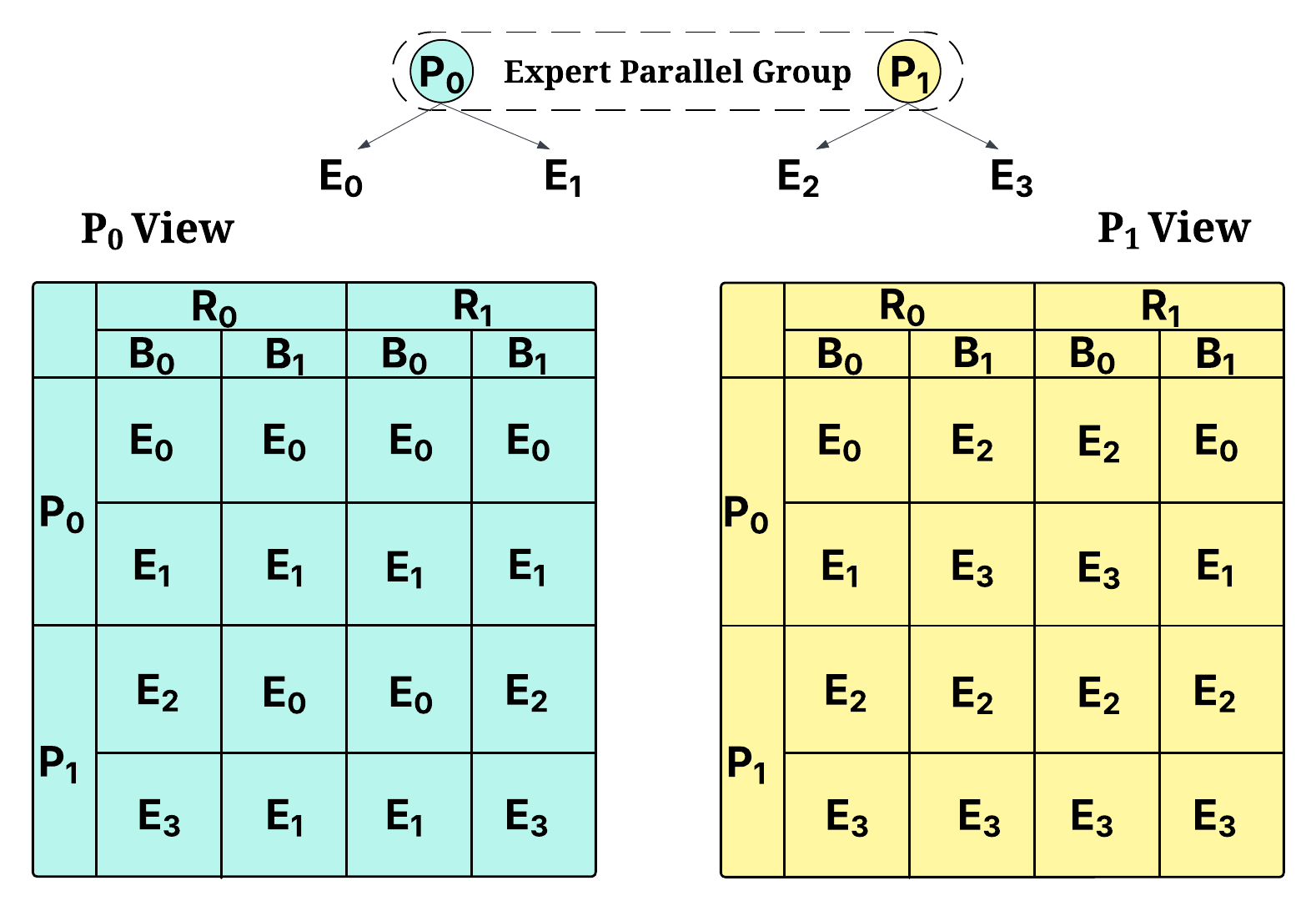}
        \caption{\emph{Layout across 2 Expert-parallel Processes}.}
        \label{fig:mem_layout}
    \end{subfigure}
    \begin{subfigure}{0.325\textwidth}
        \centering
        \includegraphics[width=\linewidth, keepaspectratio]{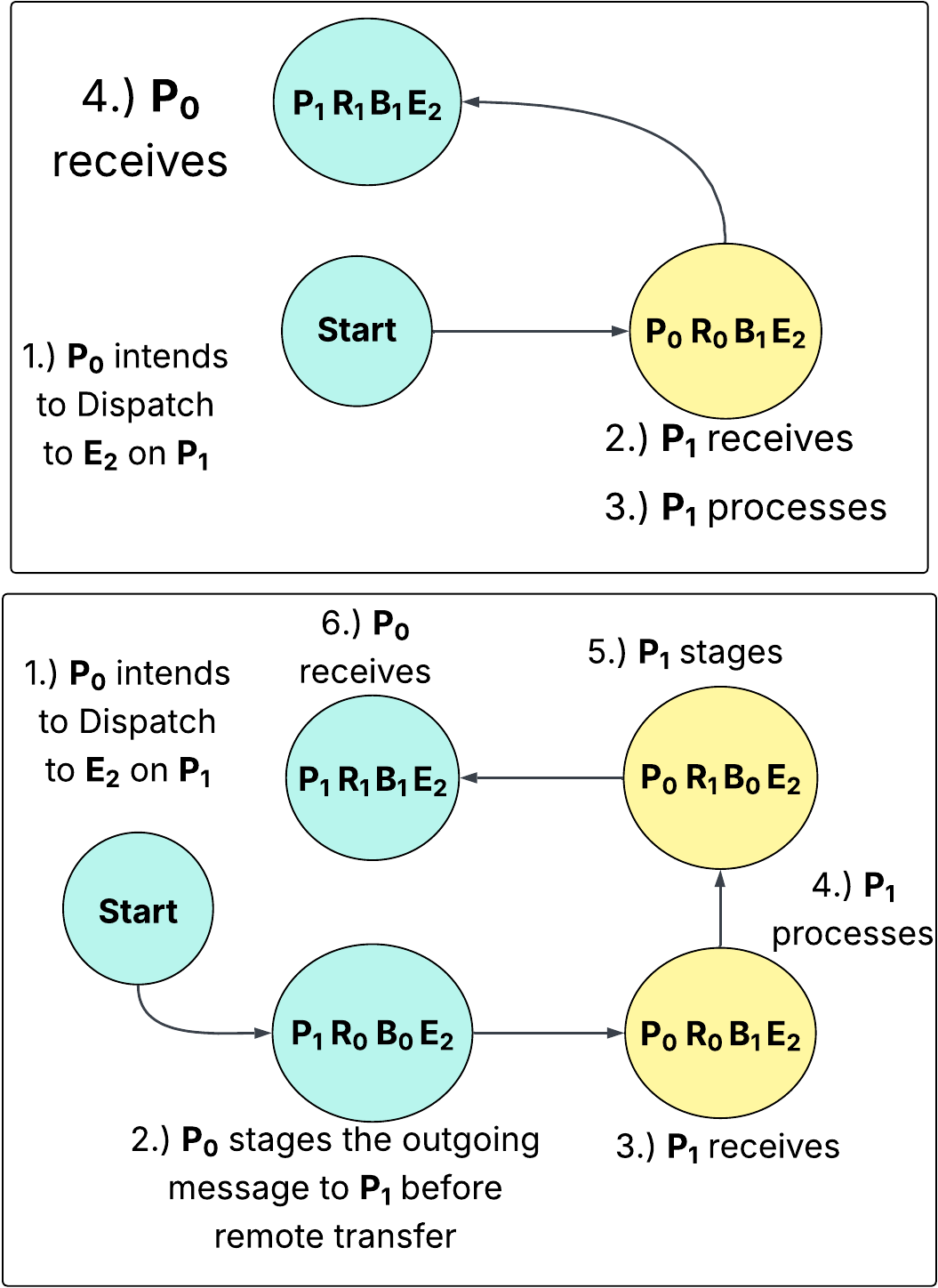}
        \caption{\emph{State machine for DMA (top) and RDMA (bottom) communication.}}
        \label{fig:sm}
    \end{subfigure}\caption{Symmetric Tensor Layout}
    \label{fig:symmt}
\end{figure}

Within a single GPU device, the actors in \sysname communicate through the GPU's memory subsystem.
Specifically, the Scheduler and Subscriber actors exchange data via fast shared memory, while other actor pairs
communicate through global memory.
For communication across multiple devices, \sysname uses \emph{device-initiated communication},
leveraging the one-sided PGAS (Partitioned Global Address Space) programming model~\cite{10.1145/1278177.1278183}.
However, achieving scalable and correct one-sided memory accesses in PGAS without costly synchronization
is a known challenge~\cite{deepep, triton-dist}.
We address this challenge with a provably correct and scalable solution: a symmetric tensor layout $L$,
supporting fully non-blocking memory accesses.
We define L as:\\
\[
    L \in \mathbb{R}^{P\times R \times B \times E \times C \times H}
\]
where: $P$ is the expert parallel world size, $R$ identifies communication rounds (\ie two rounds,
one for token dispatch and one for combine), $B$ is number of staging buffers,
$E$ is the number of local experts, $C$ is the upscaled expert capacity (\S\ref{subsubsec:payload})
and $H$ is the token embedding dimension.
Our core insight to enable non-blocking communication is \emph{temporal buffering}.
Specifically, we overprovision memory for the underlying token matrix by at least $2 \cdot r$ times, where $r$
is the number of communication rounds in the dependency graph, and the factor of $2$ accounts for
separate buffers for incoming and outgoing data within each communication round.
For MoE models, we have $2 \cdot r = 4$.
This modest increase in memory usage eliminates the need for synchronization during one-sided data transfers.
Figure~\ref{fig:sm} illustrates how cells within this symmetric tensor layout are indexed
and used for Direct Memory Access (DMA) and Remote DMA (RDMA) operations.
As Theorem~\ref{theorem:ww} reinforces,
this indexing scheme over $L$ is the underlying mechanism that allows for fully non-blocking accesses eliding
synchronization because all accesses are write \emph{conflict-free}.
See\S~\ref{sec:proof-of-theorem} for the proof.
\begin{theorem}\label{theorem:ww}
   The symmetric tensor layout $L$ is write-write conflict-free.
\end{theorem}
To construct $L$, we start from the original token buffer $T \in \mathbb{R}^{S \times H}$, where $S$ is the
sequence length and $H$ is the token embedding dimension.
We first reorganize the sequence dimension $S$ into three sub-dimensions representing the expert capacity ($C$),
local expert slots ($E$), and the expert parallel world size ($W$), st:
\[
C \cdot E \cdot W = C \cdot E' = S', \quad\text{where}\quad S' \geq S \text{ and } E' \geq E_W
\]
In the typical case of uniform expert distribution (illustrated in Figure~\ref{fig:mem_layout}),
we have $S' = S$ and $E' = E_W$, where $E_W$ is the total number of experts in the model.
Thus, the size of the token buffer is $Size(T) = S' \cdot H$.
In Figure~\ref{fig:mem_layout}, each cell labeled $E_i$ (with $i \in \{0,\ldots,3\}$) is a matrix of size $(C, H)$.
Extending prior work~\cite{DBLP:conf/iclr/LepikhinLXCFHKS21, comet}, we introduce additional temporal dimensions
$R$ (communication rounds) and $B$ (staging buffers).
Each communication round has two fixed staging slots: one for outgoing tokens and another for incoming tokens.
Each slot, indexed by dimension $P$, forms a tensor of shape $(S', H)$.
Therefore, the tensor size $Size(L)$ is generally at least four times the original token buffer size,
becoming exactly four times larger in the case of uniform expert distribution.
Empirically, we find $Size(L) \approx 4 \cdot Size(T)$, contributing memory overhead $\leq 2$\% of memory capacity for
inference of popular models.
We present a thorough breakdown in~\S\ref{sec:eval:memory}.

\subsubsection{In-place Padding for Payload Efficiency}\label{subsubsec:payload}
Due to the dynamic and uneven distribution of tokens in MoE dispatch~\cite{bmamba}, GPUs commonly
receive fewer tokens than their predefined expert capacity.
Current MoE frameworks~\cite{pmlr-v162-rajbhandari22a} typically pad these buffers with null tokens before computation,
unnecessarily increasing communication payloads and degrading performance.
In contrast, we propose \emph{in-place padding}, performing padding directly within the local
symmetric tensor buffers and thus eliminating excess network communication.

As we show in Figure~\ref{fig:mem_layout} as a reference, each cell $E_i$ is sized according to the expert capacity $C$.
We further align this capacity to ensure divisibility by the tile block size $bM = 128$,
guaranteeing safe and aligned memory reads by Processor threads consuming remote tokens.
This in-place padding strategy slightly increases the memory footprint of $L$, as described below:
\[
    Size(L) \approx \begin{cases}
        4 \cdot Size(T), & \frac{S}{E} \geq bM \\[1ex]
        4 \cdot \frac{bM \cdot E}{S} \cdot Size(T), & \text{otherwise}
    \end{cases}
\]

%% file: content/evaluation-v2.tex
\section{Evaluation}
\label{sec:evaluation}
We implement (\S\ref{sec:implementation}) and evaluate \sysname across
five metrics: \textbf{Forward Latency} (\S~\ref{subsec:forward-latency}),
\textbf{GPU Utilization} (\S~\ref{subsec:gpu-utilization}),
\textbf{Overlap Efficiency} (\S~\ref{subsec:overlap-efficiency}),
\textbf{Throughput} (\S~\ref{subsec:throughput}), and \textbf{Expert Scalability} (\S~\ref{subsec:expert-scalability}).
We run experiments on a server with 8 NVIDIA H100 80G GPUs interconnected via NVLink,
125 GB of RAM, and 20 vCPUs. We used PyTorch 2.6.0, CUDA 12.8, and Ubuntu 22.04.
All experiments use MoE transformer models configured with 16 attention heads,
an embedding dimension of 2048, and an FFN intermediate size of 2048.
We apply Distributed Data Parallelism (DDP) and Expert Parallelism for all experiments.
We execute only the forward pass over a single MoE layer and measure the average runtime
of 32 passes after 32 warmup passes.
We use top-2 routing with a capacity factor of 1.0.
We compare \sysname against several state-of-the-art MoE systems:
(1) \textbf{Comet}~\cite{comet}, 
(2) \textbf{FasterMoE}~\cite{fastermoe}, 
(3) \textbf{Megatron-CUTLASS}~\cite{megatron}, and
(4) \textbf{Megatron-TE}: Megatron-LM with Transformer Engine~\cite{transformer-engine}.
Comet relies on~\verb|cudaMemcpyPeerAsync|~\cite{fluxp2p}, while FasterMoE and Megatron-LM use NCCL exclusively for communication.

\myparab{Desiderata.}
We observe Comet exhibiting anomalously bad performance values at 8 GPUs,
so we exclude their results from evaluations at 8 GPUs and only include for results at $\leq$
4 GPUs. We evaluate \sysname using FP32 precision whereas all baselines use FP16.
We do so because (1) of incomplete fp16 tuning (\S\ref{sec:fp16-memory-throughput})
and (2) no baseline supports FP32.
Note, this precision discrepancy disadvantages \sysname by doubling its
communication volume and computational workload.
\subsection{Forward Latency}\label{subsec:forward-latency}
\begin{figure}[!h]
    \centering
    \begin{subfigure}{0.49\textwidth}
        \centering
        \includegraphics[width=\linewidth, keepaspectratio]{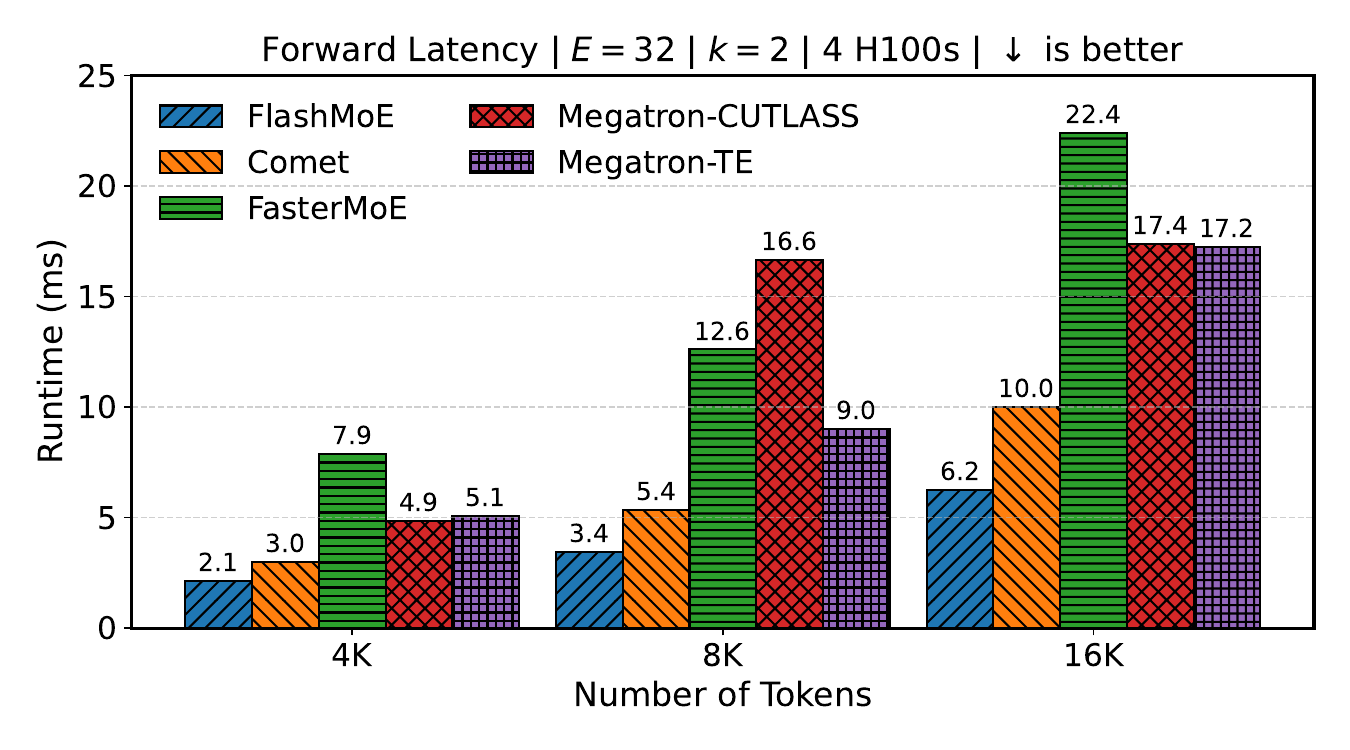}
        \caption{4 H100s}
        \label{sub:4gl}
    \end{subfigure}
    \begin{subfigure}{0.49\textwidth}
        \centering
        \includegraphics[width=\linewidth, keepaspectratio]{flash_figs/scaling_tokens_8}
        \caption{8 H100s}
        \label{sub:8gl}
    \end{subfigure}
    \caption{Forward Latency as the \emph{Number of Tokens} per GPU increases.}
    \label{fig:fl}
\end{figure}
We first measure the forward latency of \sysname across different sequence lengths on both 4 and 8 GPU setups
(Figure~\ref{fig:fl}).
\sysname consistently outperforms all baselines,
with especially notable improvements at longer sequence lengths.
On 4 GPUs, it achieves up to \textbf{4.6}x speedup over Megatron-TE at 16K tokens,
and \textbf{2.6}x over FasterMoE.
The gains are even more pronounced at 8 GPUs
where \sysname maintains low latency, exhibiting up to \textbf{6.4}x speedup over baselines that
degrade steeply due to increasing communication costs as token buffers increase proportionally.
\subsection{GPU Utilization}\label{subsec:gpu-utilization}
\begin{wrapfigure}{r}{0.5\textwidth}
    \vspace{-10pt}
    \centering
    \includegraphics[width=0.9\linewidth, keepaspectratio]{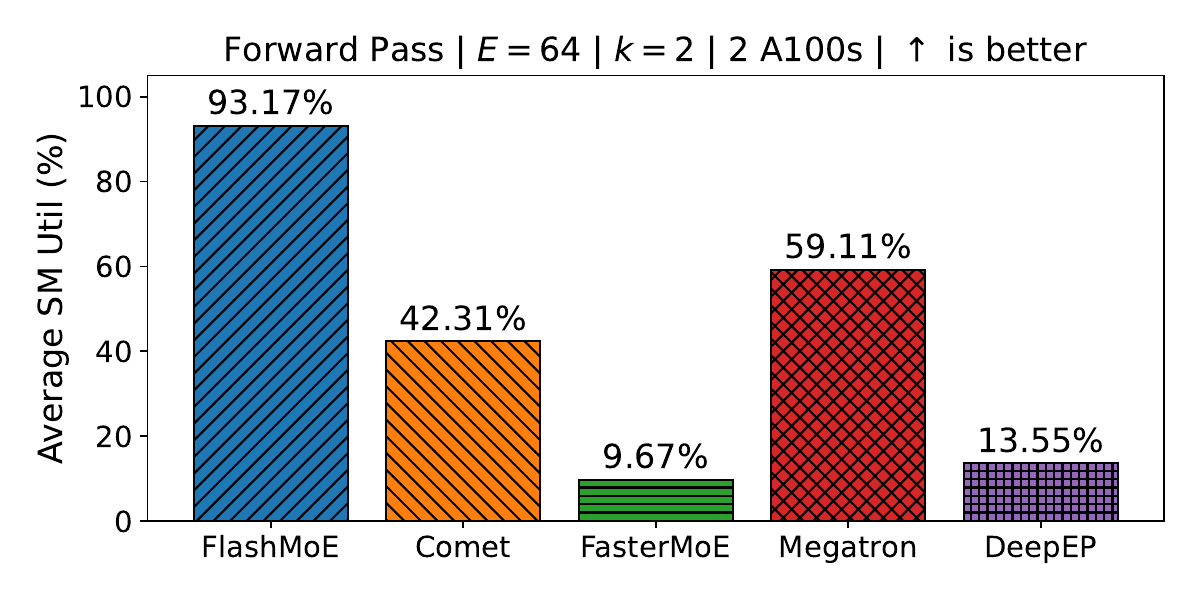}
    \caption{SM utilization, defined as the ratio of cycles in which SMs
    have at least one warp in flight
    to the total number of cycles~\cite{nsight-metrics}.
    Values represent the average SM utilization over 100 iterations.}
    \label{fig:smu}
\end{wrapfigure}
To quantify GPU efficiency, we measure Streaming Multiprocessor (SM) utilization during the forward pass (Figure~\ref{fig:smu}).
\sysname achieves 93.17\% average SM utilization,
over \textbf{9}x higher than FasterMoE (9.67\%), \textbf{6.8}x higher than DeepEP+Megatron-LM (13.55\%)
\textbf{4}x higher than Megatron-TE (59.11\%), and
\textbf{2.2}x higher than Comet (42.31\%).
This improvement stems from our fully fused kernel architecture and
fine-grained pipelining of compute and communication tasks.
By eliminating idle gaps due to kernel launches and enabling in-kernel task scheduling,
\sysname ensures SMs remain busy with productive work throughout execution.
\subsection{Throughput}\label{subsec:throughput}
\begin{wrapfigure}{r}{0.6\textwidth}
    \vspace{-15pt}
    \centering
    \includegraphics[width=0.9\linewidth, keepaspectratio]{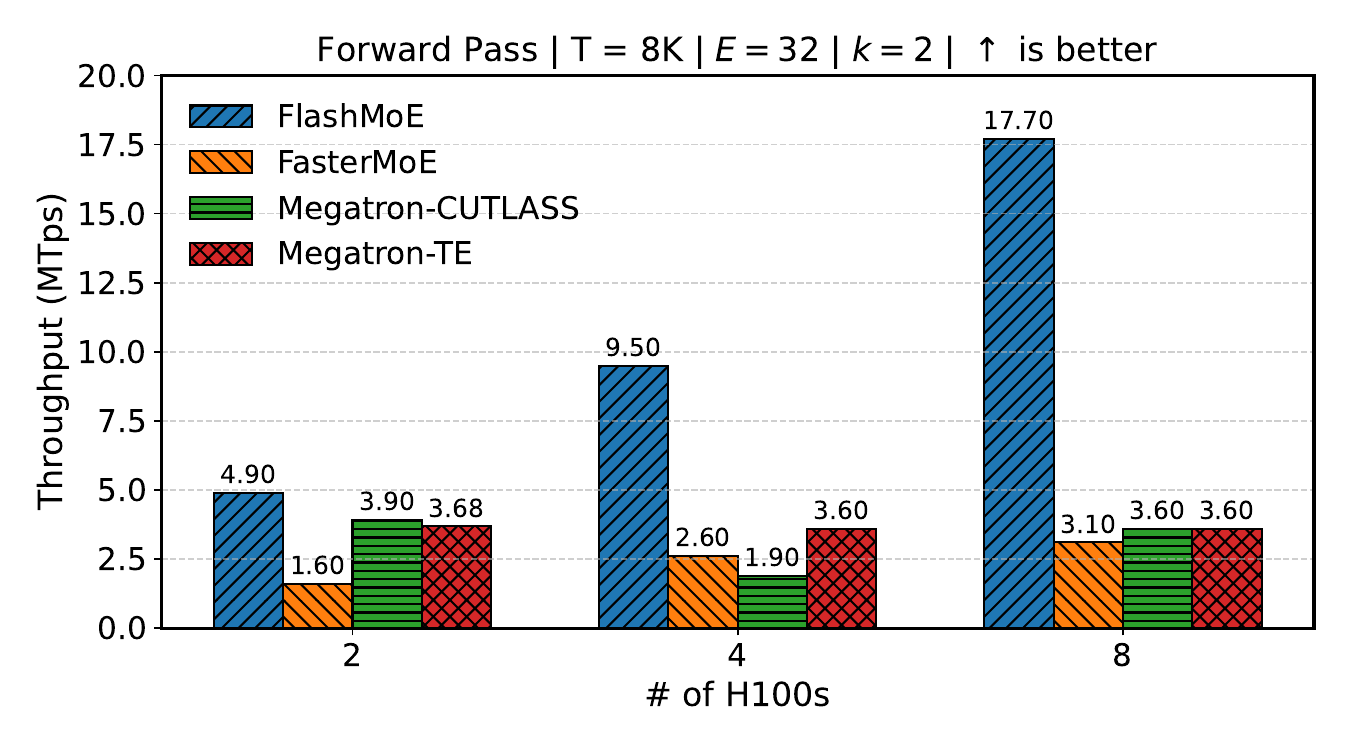}
    \caption{Throughput when scaling the number of GPUs, computed as $\frac{T \times N_G}{\text{latency}}$.}
    \label{fig:thr}
\end{wrapfigure}
As shown in Figure~\ref{fig:thr}, \sysname scales linearly with GPU count, reaching 17.7 MTokens/s at 8 GPUs.
This is over \textbf{5.7}x higher than FasterMoE and \textbf{4.9}x higher than Megatron-TE and Megatron-CUTLASS\@.
Notably, these results are achieved despite \emph{\sysname operating entirely in FP32,
    while baselines use FP16}.
This indicates that \sysname’s design eliminates throughput bottlenecks not by
exploiting lower precision, but by maximizing hardware utilization and eliminating host-driven inefficiencies.
\subsection{Overlap Efficiency}\label{subsec:overlap-efficiency}
\begin{figure}[!h]
    \centering
    \begin{subfigure}{0.49\textwidth}
        \centering
        \includegraphics[width=\linewidth, keepaspectratio]{flash_figs/scaling_gpus_8}
        \caption{Latency as Number of GPUs increases.}
        \label{fig:lng}
    \end{subfigure}
    \begin{subfigure}{0.49\textwidth}
        \centering
        \includegraphics[width=\linewidth, keepaspectratio]{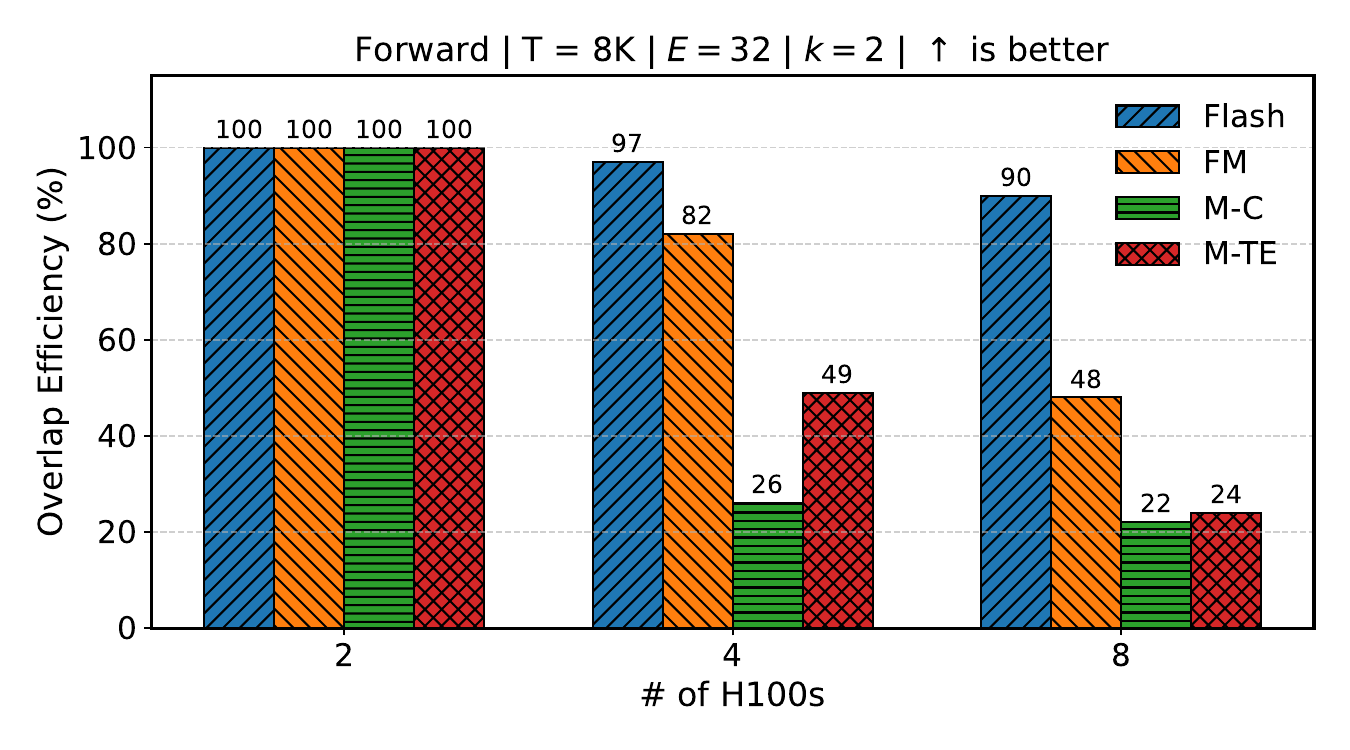}
        \caption{Weak scaling efficiency}
        \label{fig:oe}
    \end{subfigure}
    \caption{Weak scaling efficiency. We define Overlap Efficiency $O_e$
        to be $O_e = T(2) / T(N_G)$, where $T(N_G)$ is the latency at $N_G$ GPUs and $T(2)$ is the latency at 2 GPUs.}
    \label{fig:oet}
\end{figure}
We evaluate the extent to which \sysname overlaps communication and computation by measuring weak scaling efficiency
as the number of GPUs increases (Figure~\ref{fig:oe}).
We note that most baselines fail to execute at a single GPU, hence why we use 2 GPUs as the reference point.
We observe that Megatron-CUTLASS and Megatron-TE degrade significantly,
with overlap efficiency dropping below 50\% at $\geq 4$ GPUs. \sysname gives up to \textbf{3.88}x and
\textbf{4}x higher efficiency at 4 and 8 GPUs, respectively.
Figure~\ref{fig:lng} further illuminates this efficiency, as \sysname shows stable forward latency growth.
These results corroborate that \sysname's actor-based design and asynchronous data movement
achieve near-ideal overlap.
\subsection{Expert Scalability}\label{subsec:expert-scalability}
\begin{figure}[!h]
    \centering
    \begin{subfigure}{0.49\textwidth}
        \centering
        \includegraphics[width=\linewidth, keepaspectratio]{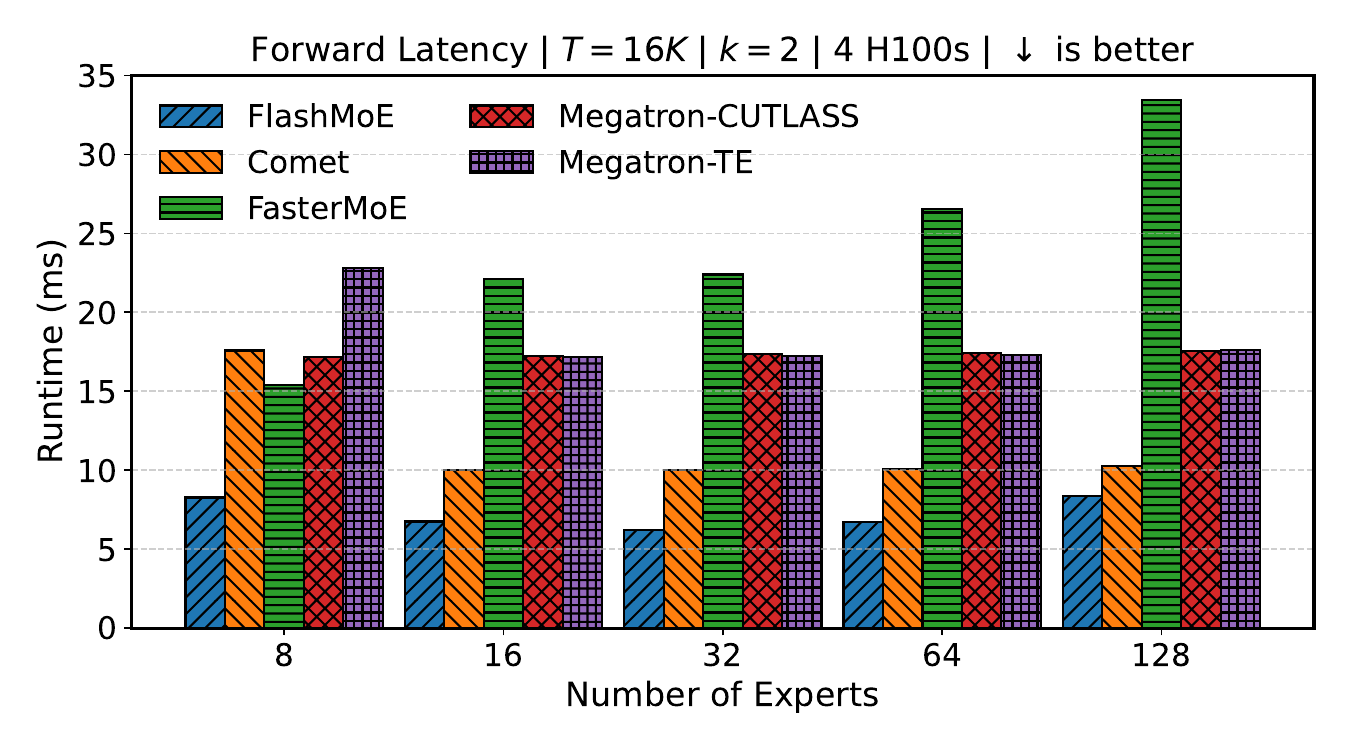}
        \caption{4 H100s}
        \label{sub:4gx}
    \end{subfigure}
    \begin{subfigure}{0.49\textwidth}
        \centering
        \includegraphics[width=\linewidth, keepaspectratio]{flash_figs/scaling_experts_8}
        \caption{8 H100s}
        \label{sub:8gx}
    \end{subfigure}
    \caption{Forward Latency as the \emph{Number of experts} increases.}
    \label{fig:xs}
\end{figure}
We analyze how \sysname scales with increasing number of experts at fixed sequence length (T = 16K).
Note that for the discussed plots, the number of experts on the x-axis is the \emph{total number across all GPUs}.
Each GPU gets 1/8th of this value.
As seen in Figure~\ref{fig:xs}, \sysname maintains \emph{low, uniform} latency, as desired,
even as the number of experts grows from 8 to 128.
In contrast, baselines exhibit superlinear latency increases due to increased kernel launch overheads.
\sysname outperforms these baselines by up to \textbf{4}X at 4 H100s and \textbf{6.6}X at 8 H100s, both at 128 experts.
\sysname’s payload-efficient communication and scheduler-driven
in-kernel dispatching allow it to sustain expert parallelism
without incurring the communication and orchestration penalties seen in other systems.
These results reinforce \sysname’s scalability for ultra-sparse MoE configurations.

%% file: content/limitations.tex
\section{Limitations and Future Work}\label{sec:limitations}
\myparab{Engineering complexity.} Fully fused, persistent kernels demand deep GPU + distributed-systems expertise;
future work may investigate compiler/DSL abstractions to lower this barrier.

\myparab{FP16 inefficiency.} Our FP16 path is suboptimal (\S\ref{sec:fp16-memory-throughput})
due to insufficient tuning.
We anticipate addressing this gap with autotuned GEMM operators like cuBLASDx~\cite{cudx} or
CUTLASS builders.

\myparab{Training support.} This work targets inference;
enabling training will require fusing backward computation and gradient communication with new bookkeeping and
task descriptors.

%% file: content/conclusion.tex
\section{Conclusion}\label{sec:conclusion-and-future-work}
We introduce \sysname, the first work to fuse the entire Distributed MoE operator into a
single persistent GPU kernel that unifies computation,
communication, and scheduling via actor-style concurrency, warp specialization, and async (R)DMA\@.
We address two dominant bottlenecks in prior systems—CPU-managed synchronous communication and fragmented multi-kernel execution.
Empirically, FlashMoE achieves up to \textbf{6×} speedup, \textbf{9×} higher GPU utilization, and \textbf{5.7×}
throughput for distributed MoE. Looking ahead, we see a shift from CPU orchestration to fully autonomous,
GPU-native pipelines—extending this fusion approach to training and beyond.

%% file: content/acks.tex
\section{Acknowledgements}\label{sec:acknowledgements}
This research is supported by NSF Award \#2444537 and ACE, one of the seven centers in JUMP 2.0,
a Semiconductor Research Corporation (SRC) program sponsored by DARPA.
This work also used resources of the
National Energy Research Scientific Computing Center,
a DOE Office of Science User Facility supported by the Office of Science of the U.S.
Department of Energy under Contract No. DE-AC02-05CH11231 using NERSC award ASCR-ERCAP0030076.
We acknowledge and thank Dr. Giulia Guidi for providing access to these NERSC supercomputing resources.

%% file: content/checklist.tex

\newpage
\section*{NeurIPS Paper Checklist}
\begin{enumerate}

    \item {\bf Claims}
    \item[] Question: Do the main claims made in the abstract and introduction accurately reflect the paper's contributions and scope?
    \item[] Answer: \answerYes{} 

    \item {\bf Limitations}
    \item[] Question: Does the paper discuss the limitations of the work performed by the authors?
    \item[] Answer: \answerYes{} 

    \item {\bf Theory assumptions and proofs}
    \item[] Question: For each theoretical result, does the paper provide the full set of assumptions and a complete (and correct) proof?
    \item[] Answer: \answerYes{} 

    \item {\bf Experimental result reproducibility}
    \item[] Question: Does the paper fully disclose all the information needed to reproduce the main experimental results of the paper to the extent that it affects the main claims and/or conclusions of the paper (regardless of whether the code and data are provided or not)?
    \item[] Answer: \answerYes{} 

    \item {\bf Open access to data and code}
    \item[] Question: Does the paper provide open access to the data and code, with sufficient instructions to faithfully reproduce the main experimental results, as described in supplemental material?
    \item[] Answer: \answerYes{} 

    \item {\bf Experimental setting/details}
    \item[] Question: Does the paper specify all the training and test details (e.g., data splits, hyperparameters, how they were chosen, type of optimizer, etc.) necessary to understand the results?
    \item[] Answer: \answerNA{} 

    \item {\bf Experiment statistical significance}
    \item[] Question: Does the paper report error bars suitably and correctly defined or other appropriate information about the statistical significance of the experiments?
    \item[] Answer: \answerYes{} 
    \item[] Justification: All reported results in the evaluation section were obtained as the average of 32 executions preceded by 32 warmup runs.

    \item {\bf Experiments compute resources}
    \item[] Question: For each experiment, does the paper provide sufficient information on the computer resources (type of compute workers, memory, time of execution) needed to reproduce the experiments?
    \item[] Answer: \answerYes{} 

    \item {\bf Code of ethics}
    \item[] Question: Does the research conducted in the paper conform, in every respect, with the NeurIPS Code of Ethics \url{https://neurips.cc/public/EthicsGuidelines}?
    \item[] Answer: \answerYes{} 

    \item {\bf Broader impacts}
    \item[] Question: Does the paper discuss both potential positive societal impacts and negative societal impacts of the work performed?
    \item[] Answer: \answerNA{} 
    \item[] Justification: We do not foresee immediate social or ethical impacts, but we acknowledge that increased compute efficiency could amplify access to large-scale models, which raises general considerations around prevalent issues such as environmental cost of training, and responsible downstream use. We recommend that users of our system consider these factors when integrating it into broader ML applications.

    \item {\bf Safeguards}
    \item[] Question: Does the paper describe safeguards that have been put in place for responsible release of data or models that have a high risk for misuse (e.g., pretrained language models, image generators, or scraped datasets)?
    \item[] Answer: \answerNA{} 

    \item {\bf Licenses for existing assets}
    \item[] Question: Are the creators or original owners of assets (e.g., code, data, models), used in the paper, properly credited and are the license and terms of use explicitly mentioned and properly respected?
    \item[] Answer: \answerYes{} 

    \item {\bf New assets}
    \item[] Question: Are new assets introduced in the paper well documented and is the documentation provided alongside the assets?
    \item[] Answer: \answerYes{} 

    \item {\bf Crowdsourcing and research with human subjects}
    \item[] Question: For crowdsourcing experiments and research with human subjects, does the paper include the full text of instructions given to participants and screenshots, if applicable, as well as details about compensation (if any)?
    \item[] Answer: \answerNA{} 

    \item {\bf Institutional review board (IRB) approvals or equivalent for research with human subjects}
    \item[] Question: Does the paper describe potential risks incurred by study participants, whether such risks were disclosed to the subjects, and whether Institutional Review Board (IRB) approvals (or an equivalent approval/review based on the requirements of your country or institution) were obtained?
    \item[] Answer: \answerNA{} 

    \item {\bf Declaration of LLM usage}
    \item[] Question: Does the paper describe the usage of LLMs if it is an important, original, or non-standard component of the core methods in this research? Note that if the LLM is used only for writing, editing, or formatting purposes and does not impact the core methodology, scientific rigorousness, or originality of the research, declaration is not required.
    \item[] Answer: \answerNA{} 
\end{enumerate}

%% file: content/appx/motivation.tex
\section{Supplementary Motivation}\label{sec:motivation-appx}
\begin{figure}[!ht]
    \centering
    \includegraphics[width=0.8\textwidth, keepaspectratio]{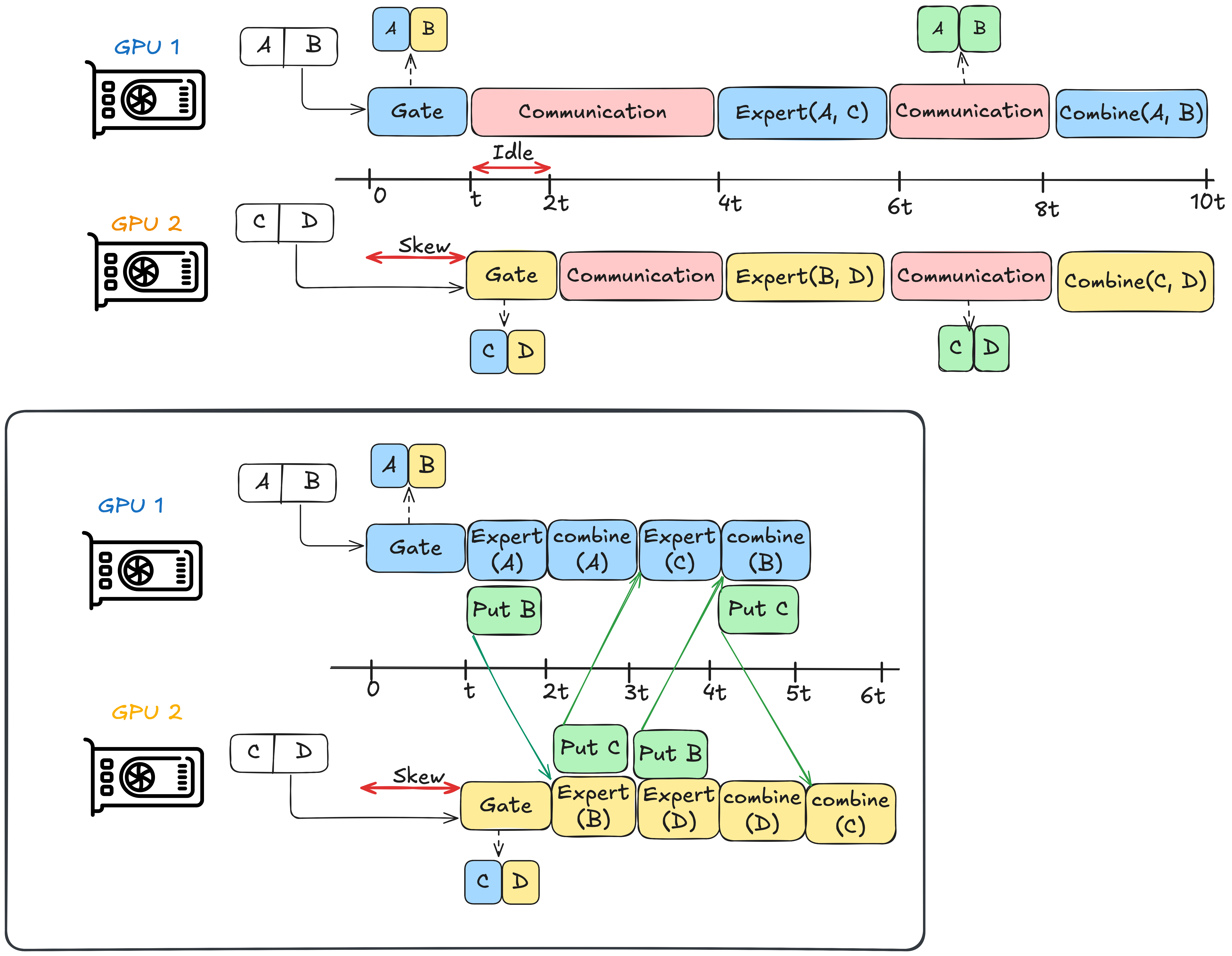}
    \caption{Overlapped Schedule (bottom) showing how idle time from the sequential schedule (top)
        is repurposed for computation. \sysname implements the overlapped schedule.}
    \label{fig:overlap}
\end{figure}
\begin{figure}[!h]
    \centering
    \begin{subfigure}{0.4\textwidth}
        \centering
        \includegraphics[width=\linewidth, keepaspectratio]{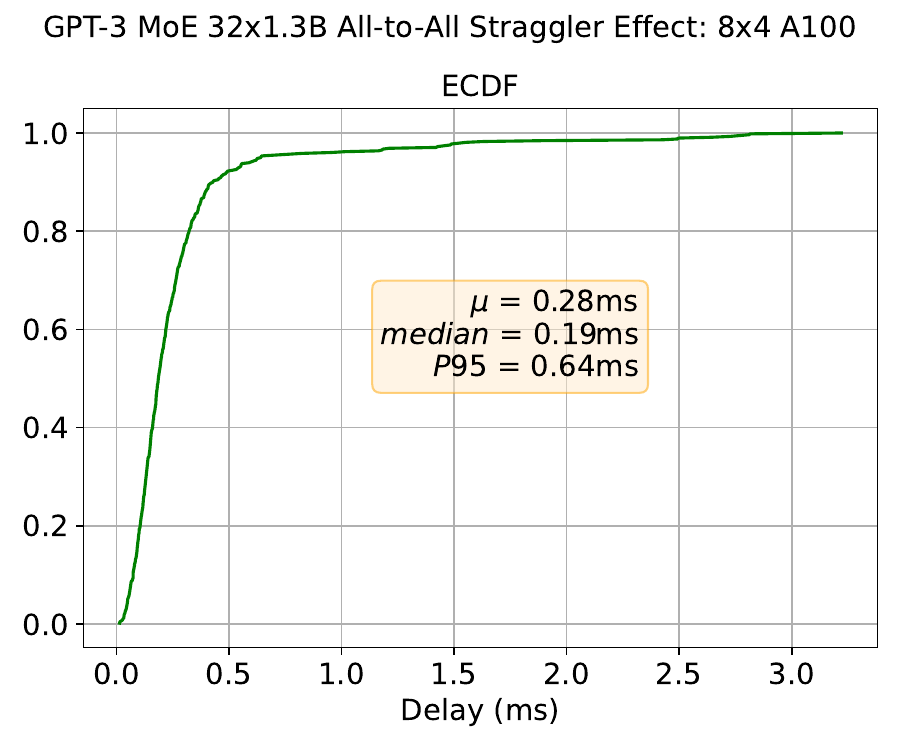}
        \caption{ECDF}
        \label{sub:ecdf_perl}
    \end{subfigure}
    \begin{subfigure}{0.4\textwidth}
        \centering
        \includegraphics[width=\linewidth, keepaspectratio]{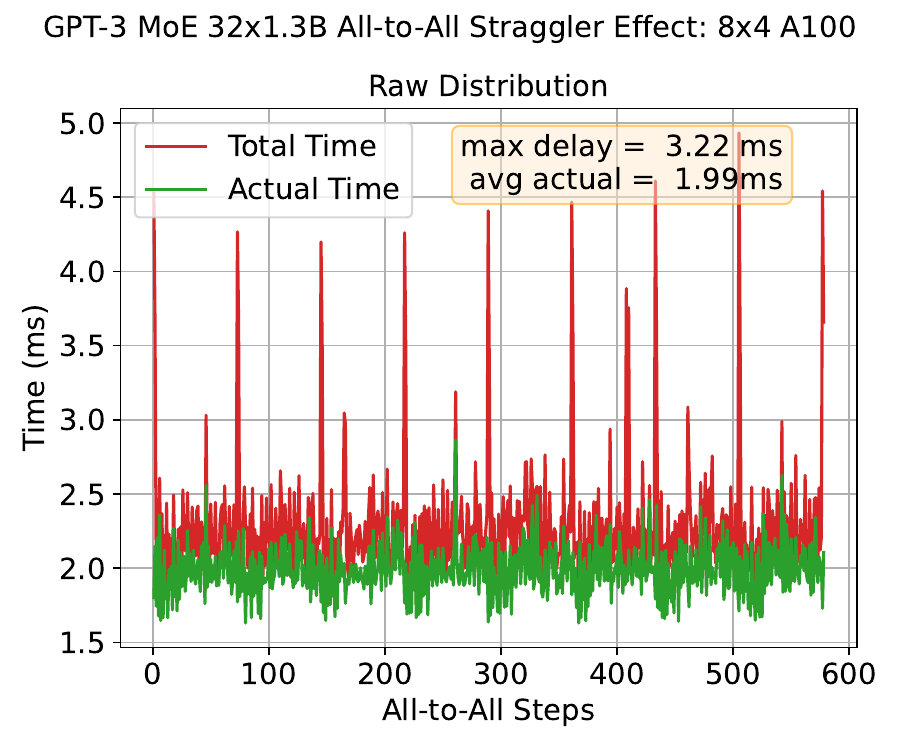}
        \caption{Raw Distribution}
        \label{sub:raw_perl}
    \end{subfigure}
    \vspace{1em} 
    \begin{subfigure}{0.4\textwidth}
        \centering
        \includegraphics[width=\linewidth, keepaspectratio]{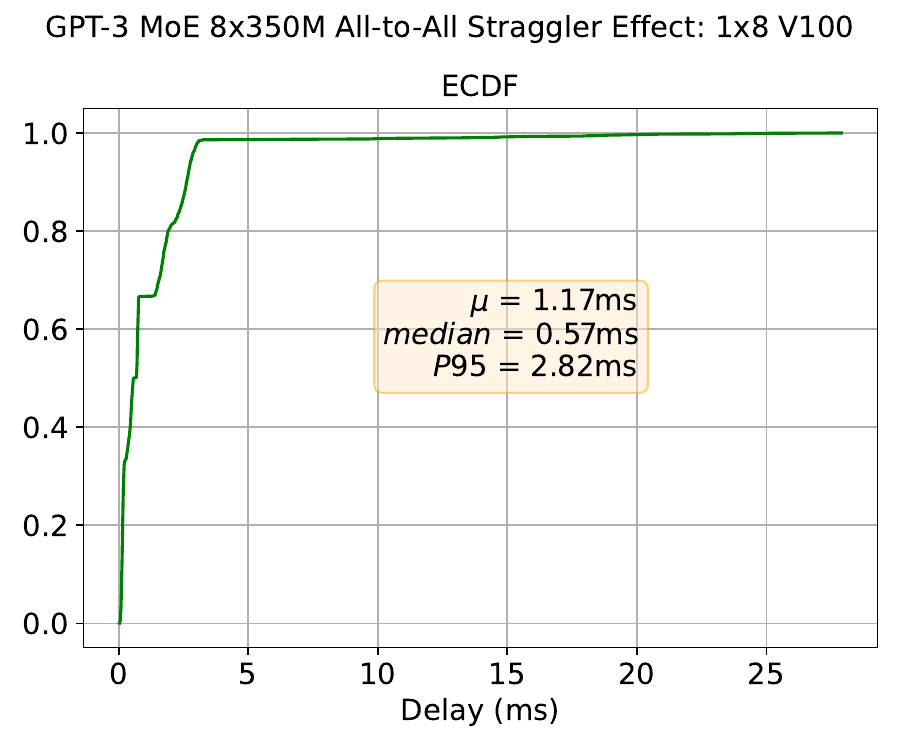}
        \caption{ECDF}
        \label{sub:ecdf_az}
    \end{subfigure}
    \begin{subfigure}{0.4\textwidth}
        \centering
        \includegraphics[width=\linewidth, keepaspectratio]{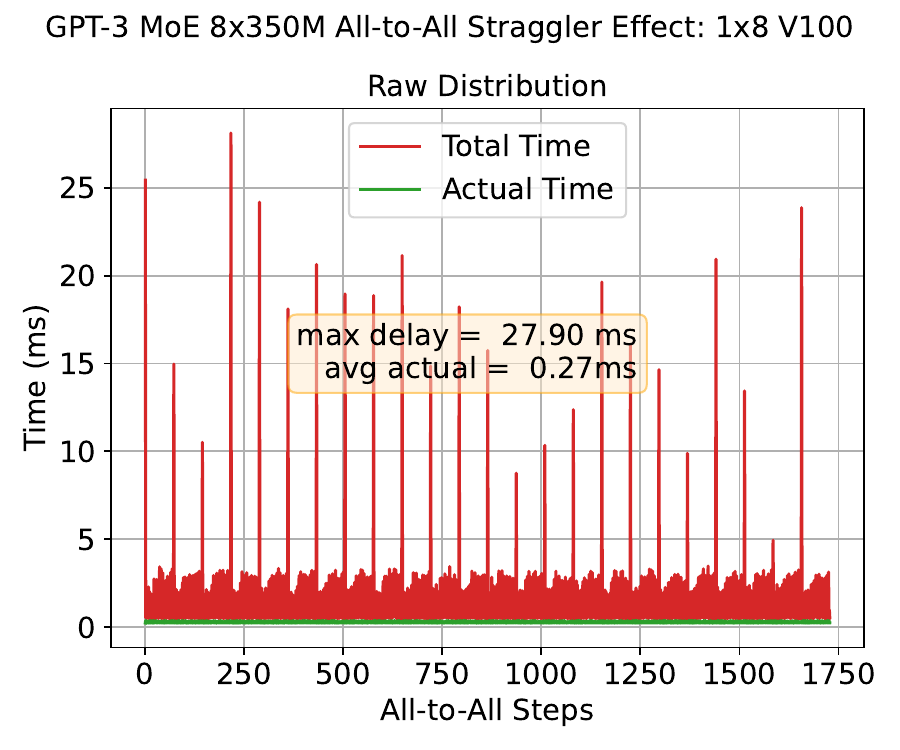}
        \caption{Raw Distribution}
        \label{sub:raw_az}
    \end{subfigure}
    \caption{Straggler effect of synchronous \alltoall. $M\times N$ A100 or V100 denotes
        $N$ GPUs within a node across $M$ nodes.
        Every GPU communicates with every other GPU per~\alltoall step.
        We capture the distribution of delay induced by stragglers across many steps.
        \textbf{Actual Time} $t_a$ denotes the fastest kernel execution time across all GPUs,
        conversely \textbf{Total Time} $t$ is the maximum recorded step time, while
        $Delay$ is the maximum difference between $t$ and $t_a$. Note $Delay$ is idle time.}
    \label{fig:straggler}
\end{figure}
In Figure~\ref{fig:straggler}, we present empirical cumulative and raw distributions of
\alltoall kernel runtime from distributed training of a 1.3B GPT-3 MoE model across
32 A100 and 8 V100 GPUs. We use this result to motivate the severity and prevalence of straggler effects.
In Figure~\ref{sub:raw_perl}, we observe P95
communication performance degradation of \textbf{1.32X} when compared to the mean actual kernel time.
This performance reduction is rather tame as the underlying hardware is a supercomputer well-tuned
against ``software jitter''~\cite{nerscNetworkNERSC}.
However, we observe a more severe p95 performance loss of \textbf{11X} in a single-node Virtual Machine (VM).
In line with prior HPC works~\cite{1639320, 10.1145/3545008.3545056},
we argue that obviating the inherent barrier in this synchronous collective communication would
allow GPUs to repurpose this observed idle time for downstream computation as depicted in Figure~\ref{fig:overlap}.
\begin{table}[!h]
    \centering
    \caption{Straggler Delay within Synchronous \emph{All-to-All} communication.
    We capture the distribution of delay induced by stragglers across many steps.
    Let \textbf{Actual Time} $t_a$ denote the fastest kernel execution time across all GPUs,
        and \textbf{Total Time} $t$ be the maximum recorded step time. We define
        $Delay$ as the maximum difference between $t$ and $t_a$. Note $Delay$ is idle time. For the
        1x8 V100, we profile 1750 steps and 600 steps for the 8x4 A100. See Figure~\ref{fig:straggler}
        for the raw distribution.}
    \label{tab:s_delays}
    \begin{tabular}{@{}lcccc@{}}
        \toprule
        \textbf{System}      & \multicolumn{1}{l}{\textbf{\# Nodes}} & \multicolumn{1}{l}{\textbf{\# GPUs}} & \textbf{Median} & \textbf{p95} \\ \midrule
        Commercial VM (V100) & 1                                     & 8                                    & 3.1x            & 11.4x        \\
        Supercomputer (A100) & 8                                     & 32                                   & 1.09x           & 1.32x        \\ \bottomrule
    \end{tabular}
\end{table}

%% file: content/related.tex
\section{Related Work}
\label{sec:related}
\noindent\textbf{Computation-Communication Overlap and Kernel Fusion.}
To reduce the communication overheads of synchronization in distributed DNN training, many research efforts have been focused on
increasing the overlap of computation and communication.
For generic Transformer-based models without MoE layers,
many works~\cite{coconet,decomposition,centauri,t3,megascale,co2,syndicate,ccfuser,fused}
have provided insights and techniques to
partition and schedule computation and communication operations, aimed at
finer-grained overlapping.
To address the challenges posed by \alltoall communication and
expert parallelism in MoE training, Tutel~\cite{tutel} and FasterMoE~\cite{fastermoe}
overlap \alltoall with expert computation.
Lancet~\cite{lancet} additionally enables both non-MoE computation in
forward pass and weight gradient computation in backward pass to be overlapped with \alltoall.
Despite overlapping, the performance of these approaches is
limited in practice due to blocking synchronous collective communication with barriers.
In contrast, \sysname fundamentally
eliminates these inefficiencies with
asynchronous, device-initiated data transfers overlapped with tiled computation
all \emph{within a single kernel}.
\sysname further differentiates itself from SOTA works like COMET~\cite{comet} and DeepEP~\cite{deepep},
which also use this form of kernel-initiated communication but at a coarse-grained granularity
and without complete kernel fusion.

%% file: content/appx/proofs.tex
\section{Proof of Theorem~\ref{theorem:ww}}\label{sec:proof-of-theorem}
We begin with two necessary definitions vital to the proof.
\begin{definition}
    Define a write as~$w(p_s, p_t, i)$, where $p_s$ is the source process and $i$ is an ordered tuple
    indicating the index coordinates for $L$ residing on the target process $p_t$.
    A write-write conflict occurs when there exist at least two distinct, un-synchronized, concurrent writes
    $w_1(p_{s_1}, p_{t_1}, i_1)$ and $w_2(p_{s_2}, p_{t_2}, i_2)$, such that
    $p_{t_1} = p_{t_2}$ and index coordinates $i_1 = i_2$ but $p_{s_1} \neq p_{s_2}$
\end{definition}
\begin{definition}
    For any source process $p_s$, a valid index coordinate $i = (p*, r, b, e, c)$ satisfies the following:
    \begin{enumerate}
        \item For inter-device writes, it must hold that $p* = p_s$ and $b = 1$.
        Note this also applies to self-looping writes $w(p_t, p_t, i)$.
        \item For any write $w(p_s, p_t, i)$, if $b = 0$, then $p_s = p_t$.
        This rule describes intra-device staging writes.
    \end{enumerate}
\end{definition}
We restate Theorem~\ref{theorem:ww} and outline its proof below.
\begin{theorem}\label{theorem:ww2}
The symmetric tensor layout $L$ is write-write conflict-free.
\end{theorem}
\begin{proof}
    As is the case for typical physical implementations,
    assume that each index coordinate $i$ maps to a distinct memory segment in $L$.
    Next, we show by contradiction that no write-write conflicts can exist when accessing $L$ using \emph{valid} $i$.
    For simplicity, we only include the index coordinates when describing a write.
    Assume that there exist at least two writes $w_1(p_{s_1}, p_{t_1}, i_1),\>w_2(p_{s_2}, p_{t_2}, i_2)$
    with $p_{t_1} = p_{t_2}$ and valid destination coordinates
    $i_1, i_2$, where $i_1 = i_2$ lexicographically and both are unpacked below.
    \[
        i_1 = (p_1, r_1, b_1, e_1, c_1),\> i_2 = (p_2, r_2, b_2, e_2, c_2)
    \]
    Note that intra-process writes always have distinct $c_j$
    coordinates, where $j \in \{0, C - 1\}$.
    For inter-process transfers, we have two cases.

    \textit{Case 1: $p_{s_1} = p_{s_2}$}
    \newline Here, $w_1$ and $w_2$ are identical operations.
    This contradicts the definition of a conflict, which requires that $p_{s_1} \neq p_{s_2}$.
    In practice, such repeat writes never even occur.

    \textit{Case 2: $p_{s_1} \neq p_{s_2}$}
    \newline To ensure validity for $i_1$ and $i_2$, it is the case that $p_1 = p_{s_1}$ and $p_2 = p_{s_2}$.
    However, this implies that $i_1 \neq i_2$ yielding a contradiction as desired.
\end{proof}

\section{Memory Overhead}\label{sec:eval:memory}
We measure the GPU memory required for the symmetric tensor $L$ and runtime bookkeeping state of \sysname.
The memory overhead primarily depends on the tile size, expert capacity ($EC$), and the number of experts ($E$).
Table~\ref{tab:memory-overhead} summarizes the memory overhead across recent MoE models~\cite{moonlight,grok,snowflake-arctic,qwen3,mixtral,deepseek} during inference, showing that \sysname maintains a modest and predictable memory footprint.
In particular, the symmetric tensor ($ST$) accounts for at most 2.15\% additional memory relative to the total inference memory requirements.
\begin{table}[!ht]
    \centering
    \caption{Memory overhead of \sysname (tile size $bM = 128$, $Size(T) = \text{Tokens} \times 4\text{KB}$).}
    \label{tab:memory-overhead}
    \small
    \setlength{\tabcolsep}{5pt}
    \renewcommand{\arraystretch}{0.9}
    \begin{tabular}{ccccccccc}
        \toprule
        \textbf{Model} & \textbf{Params} & \textbf{S} & \textbf{E} & \textbf{H} & \textbf{I} & \textbf{ST (GB)} & \textbf{Model (GB)} & \textbf{Overhead (\%)} \\
        \midrule
        Moonlight-16B-A3B  & 16B  & 8K   & 64   & 2K   & 1.38K  & 0.25  & 59    & \textbf{0.49} \\
        Grok-1             & 314B & 8K   & 8    & 6K   & 32K     & 0.75  & 1169  & \textbf{0.15} \\
        Snowflake-Arctic   & 479B & 4K   & 128  & 7K   & 4.75K   & 1.75  & 1784  & \textbf{0.12} \\
        Qwen3-235B-A22B    & 235B & 40K  & 128  & 4K   & 1.5K    & 3.00  & 875   & \textbf{0.38} \\
        Mixtral 8x7B       & 47B  & 32K  & 8    & 4K   & 14K     & 2.00  & 175   & \textbf{2.15} \\
        DeepSeek-V3        & 685B & 160K & 256  & 7K   & 2K      & 1.50  & 2551  & \textbf{0.11} \\
        \bottomrule
    \end{tabular}
    \vspace{-0.4cm}
\end{table}

%% file: content/appx/task.tex
\section{Task Implementation}\label{sec:task-implementation}
\begin{figure}[!ht]
    \centering
    \includegraphics[width=4in]{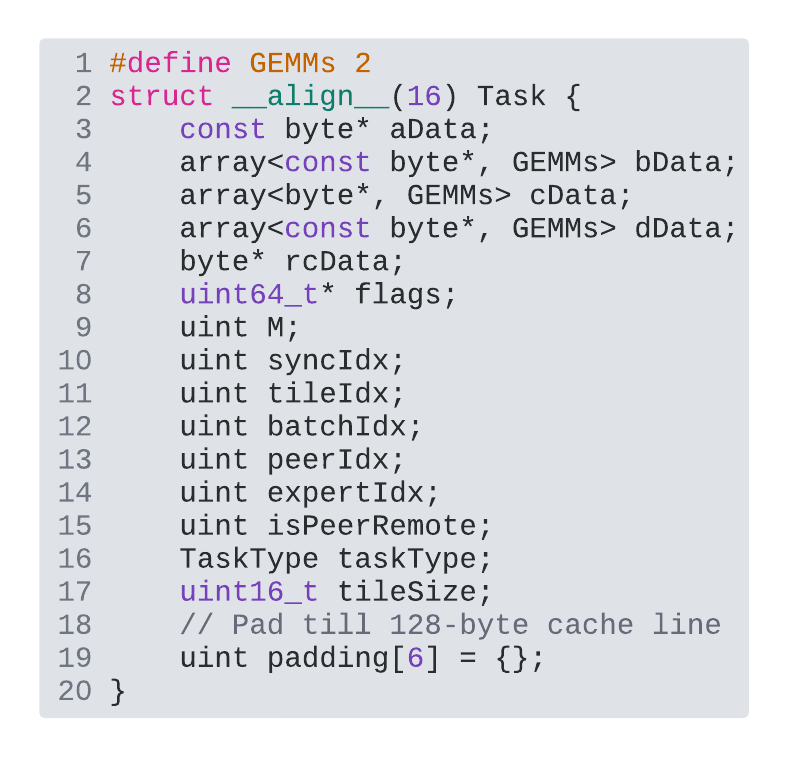}
    \caption{\emph{Task Struct}. $\text{TaskType} \in \{GEMM_0,\>GEMM_1,\>Combine\}$}
    \label{fig:task}
\end{figure}

%% file: content/appx/actors.tex
\section{Actors}\label{sec:processor-algorithm}
\subsection{Processor}\label{subsec:processor}
\begin{algorithm}[!ht]
    \DontPrintSemicolon
    \caption{~\emph{Processor Actor}: executed by a block}\label{alg:processor}
    \Begin{
        $tQ \gets \mathbf{GetTQ}()$\;
        $signal \gets 0$\;
        \tcp{shared memory variables}
        $task \gets \{\}$\;
        $interrupt \gets \kwFalse$\;
        $complete \gets \kwFalse$\;
        \While{$interrupt == $ \kwFalse}{
            \If{$warpId == 0$}{
                \If{$threadId == 0$}{
                    $\mathbf{awaitTaskFromScheduler}(interrupt,\>signal)$\;
                    $\mathbf{FencedNotifyRQ}(ready)$\;
                }
                $\mathbf{syncwarp}()$\;
                $\mathbf{warpReadTQ}(tQ, \>signal, \>task)$\;
            }
            $\mathbf{syncthreads}()$\;
            \If{$interrupt == $ \kwFalse}{
                \Switch{task.Type}{
                    \Case{$GEMM_0$}{
                        \tcp{fused GEMM, epilogue and async tile staging}
                        $\mathbf{fGET}(GEMM_0, \>task)$\;
                        \If{$threadId == 0$}{
                            $complete \gets \mathbf{NotifyTileCompletion}()$\;
                        }
                        $\mathbf{syncthreads}()$\;
                        \If{$complete == \kwTrue$}{
                            $\mathbf{NotifySchedulerNextGEMM}(tQ)$\;
                        }
                    }
                    \Case{$GEMM_1$}{
                        \tcp{fused GEMM, epilogue and async tile transfer}
                        $\mathbf{fGET}(GEMM_1, \>task)$\;
                    }
                    \Case{$Combine$}{
                        $\mathbf{combine}(task)$\;
                    }
                }
            }
        }
    }
\end{algorithm}
\clearpage
\subsection{Scheduler}\label{subsec:scheduler}
\begin{algorithm}[!ht]
    \DontPrintSemicolon
    \SetKwInput{KwInput}{Input}
    \SetKwBlock{DoParallel}{do in parallel}{end}
    \caption{~\emph{Scheduler Actor}: executed by one warp}\label{alg:scheduler}
    \Begin{
        $scheduled \gets 0$\;
        $tTB \gets 0$\;
        $tqState \gets \{\}$\;
        $pTDB \gets \mathbf{GetProcessorDoorbell}()$\;
        $sTDB \gets \mathbf{GetSubscriberDoorbell}()$\;
        $taskBound \gets \mathbf{GetTaskBound}()$\;
        $tTB \gets \mathbf{AtomicLoad}(taskBound)$\;
        \tcp{circular buffer ready queue}
        $rQ \gets \{\}$\;
        \tcp{Populate ready queue with Processor ids}
        $\mathbf{PopulateRQ}(rQ)$\;
        \While{$scheduled < tTB$}{
            $lt \gets 0$\;
            \DoParallel{
                $\text{Sweep doorbells and populate observed task counts into } tqState$\;
                $\text{Aggregate locally observed task counts into } lt$\;
            }
            $qS,\>taskTally \gets 0$\;
            \tcp{qS is the inclusive output}
            $\mathbf{WarpInclusiveSum}(lt, qS, tasktally)$\;
            \While{$tasktally > 0$}{
                $\text{Repopulate } rQ \text{ with ready processor ids}$\;
                \DoParallel{
                    $\text{Starting at } rQ[qS] \text{, signal processors about task indices from } tqState$
                }
            }
            \If{$threadId == 0$}{
                $tTB \gets \mathbf{AtomicLoad}(taskBound)$\;
            }
            $tTB \gets \mathbf{WarpBroadcast}(tTB)$
        }
        $\mathbf{InterruptSubscribers}()$\;
        $\mathbf{InterruptProcessors}()$\;
    }
\end{algorithm}
\clearpage
\subsection{Subscriber}\label{subsec:subscriber}
\begin{algorithm}[!ht]
    \DontPrintSemicolon
    \SetKwBlock{DoParallel}{do in parallel}{end}
    \SetKwInput{KwInput}{Input}
    \KwInput{$T_{\phi} \in \left(\mathbb{R}^2\right)^{E \times C}$, $G_{\phi} \in \mathbb{R}^{S \times E}$
        $O \in \mathbb{R}^{S \times H}$, $X \in \mathbb{R}^{E\times H \times D}$}
    \caption{~\emph{Subscriber Actor}: executed by three warps}\label{alg:susbcriber}
    \Begin{
        $interrupt \gets \mathbf{GetSharedInterrupt}()$\;
        $flags \gets \mathbf{GetSymmetricFlags}()$\;
        $tQ \gets \mathbf{GetTQ}()$\;
        \tcp{Predefined upper bound on the number of tasks.}
        \tcp{We modulate this value to the actual task count computed}
        \tcp{dispatch signals received from peer GPUs}
        $taskBound \gets \mathbf{GetTaskBound}()$\;
        \While{$\mathbf{AtomicLoad}(interrupt) == $ \kwFalse}{
            \tcp{dispatch flags}
            \DoParallel{
                $\text{Visit dispatch flags}$\;
                $\text{Atomically retrieve signal}$\;
                \If{$\text{Signal is set and flag is not visited}$}{
                    $\text{Mark visited}$\;
                    $\mathbf{SelfCorrectTaskBound}(taskBound, Signal)$\;
                    $\text{Enforce memory consistency before consuming packet}$\;
                    $\text{Decode packet into a set of } GEMM_0 \text{ task descriptors using } X$\;
                    $\text{Write task descriptors to } tQ$\;
                    $\text{Notify Scheduler of decoded tasks}$\;
                }
            }
            $\text{Advance flags by number of dispatch flags length}$\;
            $\text{Atomically retrieve signal}$\;
            \tcp{combine signals}
            \DoParallel{
                $\text{Visit combine flags: one per tile}$\;
                \If{$\text{Signal is set and flag is not visited}$}{
                    $\text{Mark visited}$\;
                    $\text{Enforce memory consistency before consuming packet}$\;
                    $\text{Decode packet into a set of } combine \text{ task descriptors using } T_{\phi},\> G_{\phi}, O$\;
                    $\text{Write task descriptors to } tQ$\;
                    $\text{Notify Scheduler of decoded tasks}$\;
                }
            }
        }
    }
\end{algorithm}
\clearpage

%% file: content/appx/impl.tex
\section{Implementation}\label{sec:implementation}
\begin{table}[h]
    \centering
    \caption{Implementation metrics of~\sysname.}
    \label{tab:kleos-metrics}
    \begin{tabular}{cc}
        \toprule
        \textbf{Metric} & \textbf{Value}\\ \hline
        Total lines of code (CUDA/C++) & 6820 \\
        Kernel stack frame size & 0 B \\
        Spill stores (per thread) & 0 \\
        Spill loads (per thread) & 0 \\
        Shared memory usage (per block) & 46 KB \\
        Registers per thread & 255 \\
        Max active blocks per SM & 2 \\
        Compilation time & 53 seconds \\
        Binary size & 29 MB\\
    \end{tabular}
\end{table}

%% file: content/appx/fp16_t.tex
\section{FP16 Memory Throughput}\label{sec:fp16-memory-throughput}
\begin{figure}[!ht]
    \centering
    \begin{subfigure}{\textwidth}
        \centering
        \includegraphics[width=0.8\linewidth, keepaspectratio]{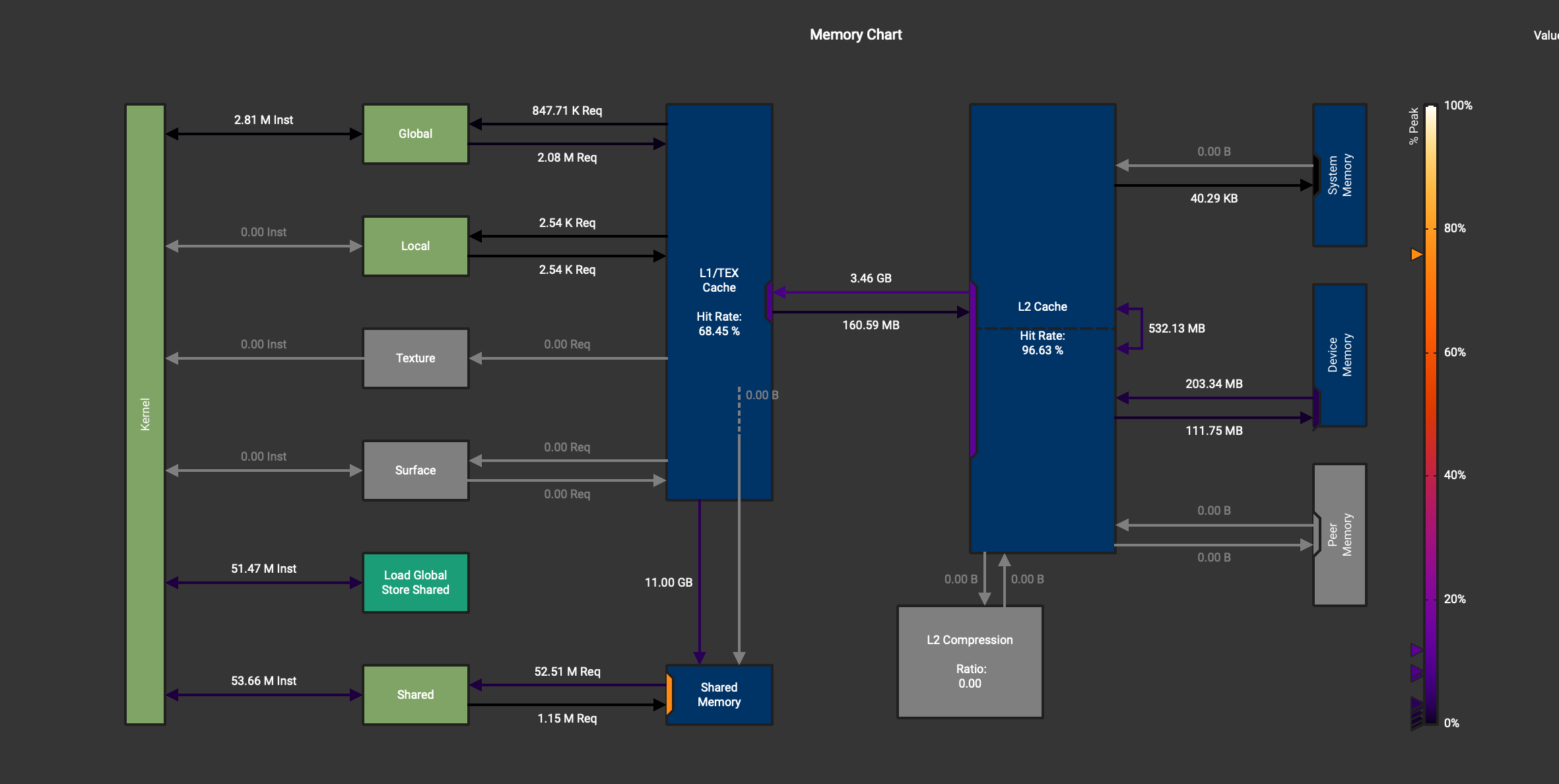}
        \caption{Memory subsystem throughput for FP16}
        \label{sub:fp16}
    \end{subfigure}
    \hfill
    \begin{subfigure}{\textwidth}
        \centering
        \includegraphics[width=0.8\linewidth, keepaspectratio]{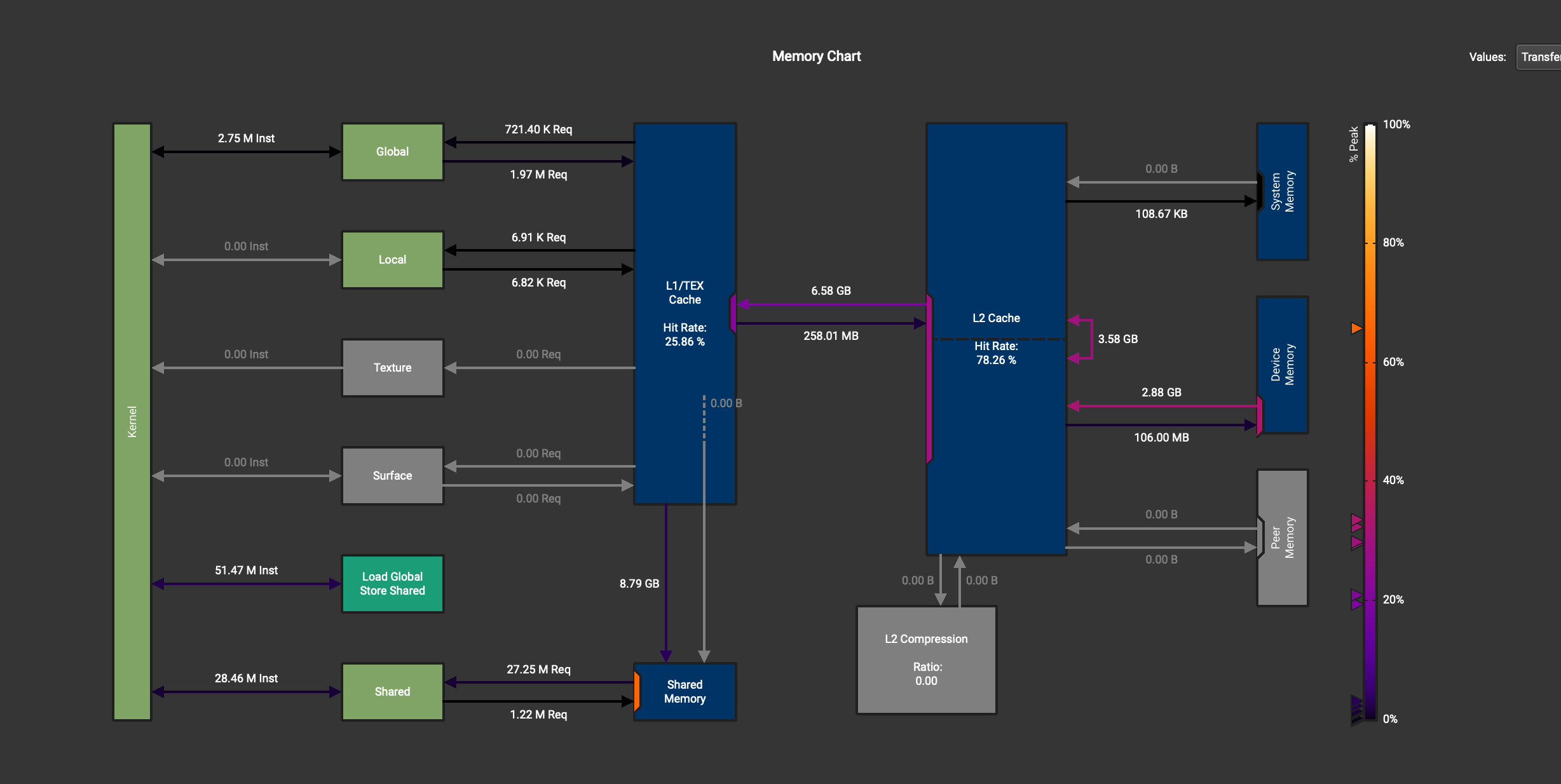}
        \caption{Memory subsystem throughput for FP32}
        \label{sub:fp32}
    \end{subfigure}
    \caption{Here, we report the total A100 memory throughput for both FP16 (top) and FP32 (bottom) variants of \sysname.
    Notably, the FP16 implementation issues approximately $2\times$
        more shared memory instructions compared to its FP32 counterpart
        under identical workloads.
        We attribute this inefficiency to
        suboptimal shared memory layouts in \sysname when
        operating on half-precision data.
        While this bottleneck is addressable through improved layout strategies,
        we leave its resolution to future work.}
    \label{fig:mem_t}
\end{figure}